\documentclass[aps,twocolumn,superscriptaddress,nofootinbib,a4paper,longbibliography]{revtex4-2}
\usepackage{graphicx, subfigure, float}
\usepackage{dcolumn}
\usepackage{bm}
\usepackage[colorlinks=true,linkcolor=blue,citecolor=magenta,urlcolor=blue]{hyperref}
\usepackage{physics}
\usepackage{amsthm}
\usepackage{bbold}
\usepackage{thm-restate}
\usepackage{amsmath, amssymb, amsfonts, amsthm, bm, bbold, physics}
\usepackage{mathtools}
\usepackage{xcolor, graphicx, subfigure, float, hyperref}
\usepackage{tcolorbox, changepage}
\usepackage{natbib}
\usepackage{dsfont}
\usepackage{thm-restate}
\usepackage[normalem]{ulem}
\usepackage{lineno}

\newtheorem{thm}{Theorem}
\newtheorem{lem}{Lemma}

\newtheorem{defi}{Definition}

\graphicspath{{figures/}}

\newtheorem*{theorem*}{Theorem}

\newtheorem*{proposition*}{Proposition}

\def\>{\rangle}
\def\<{\langle}

\begin{document}


\title{Accumulation of Device-Independent Quantum Randomness against Time-Ordered No-Signalling Adversaries}


\author{Ravishankar Ramanathan}
\email{ravi@cs.hku.hk}
\affiliation{School of Computing and Data Science, The University of Hong Kong, Pokfulam Road, Hong Kong}
\author{Yuan Liu}
\affiliation{School of Computing and Data Science, The University of Hong Kong, Pokfulam Road, Hong Kong}
\author{Yutian Wu}
\affiliation{School of Computing and Data Science, The University of Hong Kong, Pokfulam Road, Hong Kong}


\date{\today}


\onecolumngrid


\begin{abstract}
The question of security of practical device-independent protocols against no-signalling adversaries, the ultimate form of cryptographic security, has remained open. A key ingredient is to identify how the entropy in the raw outputs of a Bell test accumulates over $n$ sequential runs (termed time-ordered no-signalling) against a no-signalling adversary. Previous numerical and analytical investigations for small $n$ ($\leq 5$) had suggested that the min-entropy might not accumulate linearly in contrast to the case of quantum adversaries. Here we point out that despite the findings for small $n$, the min-entropy does in fact accumulate linearly for large $n$. We illustrate the difference in randomness accumulation against quantum and no-signalling adversaries with the paradigmatic example of the Chained Bell test for which we analytically derive the min-entropy. Finally, we illustrate the power of the no-signalling adversary by providing a class of attacks that allow an eavesdropper to perfectly guess the outputs of one player in general bipartite Pseudotelepathy games. 
\end{abstract}

\maketitle

\textit{Introduction.-}
The outcomes of local measurements on entangled quantum systems can be certified to be genuinely random through the violation of a Bell Inequality. The foundational phenomenon of Bell non-locality has given rise to the possibility of Device-Independent (DI) quantum cryptography. In the DI framework, the honest users do not need to trust even the very devices used in the cryptographic protocol. They can instead directly verify the correctness and security by means of simple statistical tests of the device's input-output behavior, specifically checking for the violation of a Bell inequality. This violation would certify that the outputs of the device could not have been pre-determined, and as such could not be perfectly guessed by an adversary, allowing for the possibility of securely extracting randomness or secret key.

Remarkably, it has been shown that in principle, security proofs for protocols of quantum key distribution (QKD) and quantum random number generation (QRNG) can be independent even of the fact that the adversary is limited by quantum theory. All that is required instead is the absence of any hidden information flow between the devices, a condition that can be enforced either by shielding the devices or by means of spacelike separation. Such security against a \textit{no-signalling} (NS) adversary would mean the achievement of the gold standard of cryptography - DI information-theoretic security based only on the principle of relativistic causality. 

However, all known schemes for such causality-based DI-QRNG and DI-QKD suffer from drawbacks: some protocols are intolerant of noise \cite{BHK05, BCK12} and hence more proof-of-principle, while the ones that do offer noise-tolerance (and hence require the use of privacy amplification as a crucial step) require multiple devices (that are either shielded or spacelike separated) with the number of devices growing with the level of security  \cite{Mas09, MRC+14, CSW16, MPA11}, and hence are also impractical. It is of course highly desirable to develop practical protocols, where two parties each use a single device carrying out measurements consecutively in a time-ordered sequence. Such a scenario is termed \textit{time-ordered no-signalling} ($\mathcal{TONS}$) whereby past measurements on a device could influence its future outcomes. In this context, several attacks and no-go results have been shown against the key privacy amplification step in the protocol \cite{AFTS12, SW16, Sal15, HRW13, SKG+21, HR19}, while the status of accumulation of randomness in the outcomes (the raw output string before privacy amplification) against such adversaries is also unclear \cite{BPA18, SW16}. Here we clarify the latter problem, namely the accumulation of randomness in the raw outcomes of the devices against $\mathcal{TONS}$ adversaries.

In pioneering work in \cite{MRC+14}, the authors considered a causal scenario termed as $\mathcal{FULL-NS}$ in which $2n$ players hold $2n$ devices that satisfy all possible pairwise no-signalling conditions (i.e., this corresponds to an impractical situation with a large number of shielded or spacelike-separated devices). For behaviors satisfying all the no-signalling constraints in this scenario, the authors showed that the guessing probability goes down exponentially (the min-entropy accumulates linearly), essentially the optimal guessing strategy for Eve corresponds to an i.i.d. (identical and independently distributed) strategy for all devices. In a comprehensive study in \cite{BPA18}, the authors investigated the guessing probability in the $\mathcal{TONS}$ scenario (as well as a few other causal scenarios), specifically focusing on the min-entropy associated with $n$ copies of noisy PR boxes (corresponding to violations of the CHSH Bell inequality). By numerically solving the guessing probability optimization problem for $n = 2,3,4,5$ runs, they found that Eve's optimal guessing probability is higher than the one obtained with the i.i.d. strategy. For cases $n=2,3$ they also found explicit solutions to the primal and dual forms of the optimization problem, thereby obtaining analytical expressions for the guessing probability. It was thereby suggested that the guessing probability may not decay exponentially with number of runs $n$ in these scenarios.


The extensive analytical and numerical investigations for small $n$ \cite{BPA18} had suggested that the randomness quantified by the min-entropy does not accumulate linearly, while not excluding such linear increase in the asymptotic limit $n \rightarrow \infty$. Our first statement is to show that in fact, the min-entropy does increase linearly for large $n$ in the TONS scenario. This result is established for all Bell inequalities which exhibit the phenomenon of monogamy of nonlocality in no-signalling theories (such games include the CHSH and the Braunstein-Caves chained Bell inequalities). Nevertheless, not all Bell inequalities are useful for randomness generation against no-signalling adversaries. 
Our second result is to show that in fact, achieving the maximal value in many two-player pseudotelepathy (PT) games does not certify randomness in the outcomes of one player against no-signalling adversaries. In other words, a large class of bipartite PT games do not exhibit monogamy of nonlocality in NS theories. As such, these games should not be used in DI randomness generation protocols against NS adversaries.  
Finally, as our third result we illustrate the power of this adversary in comparison to one limited by quantum theory, by deriving the min-entropy for the paradigmatic example of the three-input chained Bell inequality. 

\textit{Preliminaries: DI randomness and NS adversaries.-} Consider a single run of a Bell experiment in which Alice inputs $x \in \mathsf{X}$ to her device to get an output $a \in \mathsf{A}$, and similarly Bob inputs $y \in \mathsf{Y}$ into his device to get an output $b \in \mathsf{B}$. An adversary Eve holds in her possession a device with inputs $z \in \mathsf{Z}$ and outputs $e \in \mathsf{E}$, which may be correlated with the device held by Alice and Bob. The experimental setting is characterised by the conditional distributions (aka a behavior) $P_{\mathsf{A}\mathsf{B}\mathsf{E}|\mathsf{X}\mathsf{Y} \mathsf{Z}}$ which are required to satisfy the no-signalling conditions dictated by the spacetime locations of the measurements, i.e., the outputs of any subset of players is independent of the input of the remaining players to avoid signalling in each run. Eve's aim is to guess Alice's output $a$ at the end of the Bell test after getting to know Alice's input $x$. The randomness in Alice's outcome with respect to Eve is quantified in terms of the guessing probability (or equivalently its negative log termed the min-entropy). This measure of randomness finds direct application in device-independent quantum cryptography protocols for randomness certification \cite{our16, CR12, KPR+25} and key distribution \cite{MRC+14, Mas09, JMS20, V17, BHK05, BCK12}. The guessing probability is defined as the probability with which adversary can correctly guess Alice's outcomes and produce $e = a$ under the constraint that the marginal behavior of Alice-Bob $P_{\mathsf{A}\mathsf{B}|\mathsf{X}\mathsf{Y}}$ achieves value $\omega(G)$ in some non-local game $G$ with winning condition given by predicate $V(a,b,x,y) \in \{0,1\}$ and input distribution $\pi(x,y)$.


Now let us move to the scenario where Alice and Bob perform $n$ runs of the Bell experiment each instead of just a single run as above. In this case Alice, after interacting $n$ times with her device, ends up with input and output strings $\textbf{x} = \left(x_1, x_2, \ldots, x_n \right) \in \mathsf{X}^n := \mathsf{X} \times \ldots \times \mathsf{X}$ and $\textbf{a} = \left(a_1, a_2, \ldots, a_n \right) \in \mathsf{A}^n := \mathsf{A} \times \ldots \times \mathsf{A}$ respectively. Similarly, Bob obtains input and output strings $\textbf{y} \in \mathsf{Y}^n$ and $\textbf{b} \in \mathsf{B}^n$ respectively. The experimental scenario is characterised by the conditional distributions $\{P_{\mathsf{A}^n\mathsf{B}^n|\mathsf{X}^n\mathsf{Y}^n}(\textbf{a}, \textbf{b}| \textbf{x}, \textbf{y})\}$ termed a behavior. An adversary Eve has in her possession a device, with inputs $z \in \mathsf{Z}$ and outputs $e \in \mathsf{E}$, that may be correlated with the device held by Alice and Bob. In other words, the behavior of the three players is now described by $\{P_{\mathsf{A}^n\mathsf{B}^n\mathsf{E}|\mathsf{X}^n\mathsf{Y}^n\mathsf{Z}}(\textbf{a}, \textbf{b},e|\textbf{x}, \textbf{y},z)\}$ whose marginals on Alice-Bob are given by $\{P_{\mathsf{A}^n\mathsf{B}^n|\mathsf{X}^n\mathsf{Y}^n}(\textbf{a}, \textbf{b}| \textbf{x}, \textbf{y})\}$ and which we assume to not permit signalling between Eve and Alice-Bob. The adversary aims to guess Alice's output string after getting to know her input string for the $n$ runs, that is, Eve aims to achieve $e = \textbf{a}$ after receiving $\textbf{x}$.

The behavior of the device held by Alice-Bob and the conditional behaviors $P_{\mathsf{A}^n\mathsf{B}^n|\mathsf{X}^n\mathsf{Y}^n\mathsf{E}\mathsf{Z}}$ (conditioned on Eve's input-output $z, e$) obey certain constraints depending on the causal structure being considered. Here, we follow \cite{BPA18} in considering two structures of interest to DI protocols, termed $\mathcal{ABNS}$ and $\mathcal{TONS}$ defined as follows. 
\begin{widetext}
\begin{defi}
\label{def:ABNS-TONS}
The behavior $P_{\mathsf{A}^n\mathsf{B}^n|\mathsf{X}^n\mathsf{Y}^n\mathsf{E}\mathsf{Z}}$ is said to belong to the set $\mathcal{ABNS}$ if the conditional distributions obey
\begin{eqnarray}
\sum_{\textbf{b}} P_{\mathsf{A}^n\mathsf{B}^n|\mathsf{X}^n\mathsf{Y}^n\mathsf{E}\mathsf{Z}}(\textbf{a}, \textbf{b}| \textbf{x}, \textbf{y},e,z) &=& P_{\mathsf{A}^n|\mathsf{X}^n\mathsf{E}\mathsf{Z}}(\textbf{a}| \textbf{x},e,z) \quad \forall \textbf{a}, \textbf{x}, \textbf{y}, e, z, \nonumber \\
\sum_{\textbf{a}} P_{\mathsf{A}^n\mathsf{B}^n|\mathsf{X}^n\mathsf{Y}^n\mathsf{E}\mathsf{Z}}(\textbf{a}, \textbf{b}| \textbf{x}, \textbf{y},e,z) &=& P_{\mathsf{B}^n|\mathsf{Y}^n\mathsf{E}\mathsf{Z}}(\textbf{b}| \textbf{y},e,z) \quad \forall \textbf{b}, \textbf{x}, \textbf{y}, e, z.
\end{eqnarray}
The behavior $P_{\mathsf{A}^n\mathsf{B}^n|\mathsf{X}^n\mathsf{Y}^n\mathsf{E}\mathsf{Z}}$ is said to belong to the set $\mathcal{TONS}$ if the conditional distributions obey 
\begin{eqnarray}
\sum_{\textbf{b}_{>i}} P_{\mathsf{A}^n\mathsf{B}^n|\mathsf{X}^n\mathsf{Y}^n\mathsf{E}\mathsf{Z}}(\textbf{a}, \textbf{b}| \textbf{x}, \textbf{y},e,z) &=& P_{\mathsf{A}^n\mathsf{B}^{\leq i}|\mathsf{X}^n\mathsf{Y}^{\leq i}\mathsf{E}\mathsf{Z}}(\textbf{a}, \textbf{b}_{\leq i} | \textbf{x}, \textbf{y}_{\leq i}, e,z) \quad \forall \textbf{a}, \textbf{b}_{\leq i}, \textbf{x}, \textbf{y} , e, z, \nonumber \\
\sum_{\textbf{a}_{>i}} P_{\mathsf{A}^n\mathsf{B}^n|\mathsf{X}^n\mathsf{Y}^n\mathsf{E}\mathsf{Z}}(\textbf{a}, \textbf{b}| \textbf{x}, \textbf{y},e,z) &=& P_{\mathsf{A}^{\leq i}\mathsf{B}^{n}|\mathsf{X}^{\leq i}\mathsf{Y}^n\mathsf{E}\mathsf{Z}}(\textbf{a}_{\leq i}, \textbf{b} | \textbf{x}_{\leq i}, \textbf{y}, e,z) \quad \forall \textbf{a}_{\leq i}, \textbf{b}, \textbf{x}, \textbf{y} , e, z,
\end{eqnarray}
for all $i \in \{0,\ldots, n-1\}$, where $\textbf{b}_{>i} = \left(b_{i+1},\ldots, b_n \right)$, $\textbf{b}_{\leq i} = \left(b_1, \ldots, b_{i} \right)$ and similarly for the other variables. 
\end{defi}
\end{widetext}
In words, the set $\mathcal{ABNS}$ considers the causal scenario where the players' devices are shielded from each other, each player receives all their inputs at once, and produces their outputs at once independent of the inputs of the other player. It is equivalent to consider one big device on Alice's side (respectively Bob's side) that receives the input string $\textbf{x}$ and produces the output string $\textbf{a}$ - such a scenario is also referred to as parallel repetition. Note that although we have only focused on the causal constraints, the probabilities also obey non-negativity and normalisation constraints (each probability is $\geq 0$ and the probabilities of outputs $\textbf{a}, \textbf{b}$ sum to $1$ for each $\textbf{x}, \textbf{y}$, and for any $e, z$). 

The set $\mathcal{TONS}$ considers the causal scenario in which the players' devices are shielded from each other - each player's outputs do not depend on the inputs of the other player, and furthermore each player performs $n$ sequential measurements such that their outputs in prior rounds do not depend on their inputs in future rounds. Again the terms also obey the non-negativity and normalisation constraints to give rise to proper probability distributions. Evidently, we have $\mathcal{TONS} \subseteq \mathcal{ABNS}$ so that the $\mathcal{ABNS}$ gives the adversary greater power to prepare the device held by Alice-Bob. In other words, Eve's guessing probability in the $\mathcal{ABNS}$ scenario is greater than or equal to that in the $\mathcal{TONS}$ scenario. 

Every measurement of Eve can be seen as a choice of convex decomposition of Alice-Bob's behavior, i.e.,
\begin{eqnarray}
&&\sum_{e} P_{\mathsf{E}|\mathsf{Z}}(e|z) P_{\mathsf{A}^n\mathsf{B}^n|\mathsf{X}^n\mathsf{Y}^n\mathsf{E}\mathsf{Z}}(\textbf{a}, \textbf{b}| \textbf{x}, \textbf{y},e,z) \nonumber \\
&&= P_{\mathsf{A}^n\mathsf{B}^n|\mathsf{X}^n\mathsf{Y}^n}(\textbf{a}, \textbf{b}| \textbf{x}, \textbf{y}) \quad \forall \textbf{a}, \textbf{b}, \textbf{x}, \textbf{y},z.
\end{eqnarray}  
Eve's guessing probability of the outcomes for any fixed input $\textbf{x}^*$ of Alice is thus obtained by maximizing over all possible convex decompositions that are compatible with the marginal behavior achieving certain value for some non-local game $G$. The guessing probability is therefore the solution of the following linear programming problem:
\begin{widetext}
\begin{eqnarray}
\label{eq:TONS-guessprob}
P_g(\textbf{x}^*, \omega^*) = && \; \; \max_{P_{\mathsf{E}|\mathsf{Z}}, \; P_{\mathsf{A}^n\mathsf{B}^n|\mathsf{X}^n\mathsf{Y}^n\mathsf{E}\mathsf{Z}}} \; \;  \sum_{e, \textbf{b}} P_{\mathsf{E}|\mathsf{Z}}(e|z^*) P_{\mathsf{A}^n\mathsf{B}^n|\mathsf{X}^n\mathsf{Y}^n\mathsf{E}\mathsf{Z}}(\textbf{a} = e, \textbf{b}| \textbf{x}^*,\textbf{y},e,z^*) \nonumber \\
&&\text{s.t.} \; \; \; \omega^{P_{A^iB^i|X^{\leq i}Y^{\leq i}}}(G) = \omega^* \; \; \; \;  \forall i = 1,\ldots, n, \nonumber \\
&&\forall z, a_i, b_i, \textbf{x}_{\leq i}, \textbf{y}_{\leq i} \; \; P_{\mathsf{A}^i\mathsf{B}^i| \mathsf{X}^{\leq i} \mathsf{Y}^{\leq i}}(a_i, b_i| \textbf{x}_{\leq i}, \textbf{y}_{\leq i}) =  \sum_{e, \textbf{a}_{\neq i}, \textbf{b}_{\neq i}} P_{\mathsf{E}|\mathsf{Z}}(e|z) P_{\mathsf{A}^n\mathsf{B}^n|\mathsf{X}^n\mathsf{Y}^n\mathsf{E}\mathsf{Z}}(\textbf{a}, \textbf{b}| \textbf{x},\textbf{y},e,z), \nonumber \\
&& \forall z, \; \; \; \; \sum_{e} P_{\mathsf{E}|\mathsf{Z}}(e|z) = 1 \; \; \; \; \forall e, z \; \; \; P_{\mathsf{E}|\mathsf{Z}}(e|z) \geq 0 , \nonumber \\
&& \; \;  \forall \; e, z \; \; \; P_{\mathsf{A}^n\mathsf{B}^n|\mathsf{X}^n\mathsf{Y}^n, \mathsf{E} = e, \mathsf{Z} = z} \in \mathcal{TONS} 
\end{eqnarray}
\end{widetext}
Here, $\textbf{x}^*$ is a specific input string of Alice whose output Eve wants to guess, for which purpose she inputs the specific $z^*$ into her device, $\textbf{a}_{\neq i} = \left(a_1, \ldots, a_{i-1}, a_{i+1},\ldots, a_n \right)$ for any $i$, and similarly for $\textbf{b}_{\neq i}$. 


\textit{Randomness accumulation against $\mathcal{TONS}$ adversaries.-}
The guessing probability is very difficult to compute in general, as even though a linear program can be solved efficiently through simplex-based or interior point methods, the size of the linear program in \eqref{eq:TONS-guessprob} grows exponentially in $n$ (the behaviors consist of $|\mathsf{A}|^n|\mathsf{B}|^n|\mathsf{X}|^n|\mathsf{Y}|^n$ probabilities). It is useful to solve the optimisation problem for small $n$ to see if an i.i.d. strategy is always optimal, even though a linear increase in the min-entropy is not precluded if such is not the case. Specifically, in \cite{BPA18} it was found that
\begin{eqnarray}
P_g(\textbf{x}^*, \omega^*) =  \left\{
\begin{array}{ll}
      1 - \frac{3v}{4} & \; \;   n=2 \\
      1 - \frac{7v}{8} & \; \; n=3 \wedge v \geq \sqrt{5} - 2,
\end{array} 
\right. 
\end{eqnarray} 
for the CHSH Bell scenario (two binary inputs $x, y \in \{0,1\}$ per player in each round) with the marginal behavior $P_{\mathsf{A}^n \mathsf{B}^n|\mathsf{X}^n\mathsf{Y}^n} = \prod_{i=1}^n \text{PR}_v(a_i,b_i|x_i, y_i)$ where $\text{PR}_v(a_i,b_i|x_i, y_i) = (3+v)/8$ for $a_i \oplus b_i = x_i \cdot y_i$ and $(1-v)/8$ otherwise.

Here we observe that the guessing probability does go down exponentially with $n$ in the $\mathcal{TONS}$ setting for general Bell inequalities which obey a monogamy of non-locality property (which includes the CHSH inequality considered above), we defer the proof to the Appendix \ref{sec:App-TONS-rand} where we explicitly compute the constants in the exponential decrease for the three-input chained Bell inequality. 

\begin{thm}
\label{thm:guess-prob-tons-main}
Let $G$ denote a non-local game that exhibits the phenomenon of monogamy of non-locality in no-signalling theories. Consider a device-independent scenario where the marginal behavior of Alice and Bob has value $\omega^*$ for a non-local game $G$ in each of $n$ rounds. Then the probability $P_g(\textbf{x}^*, \omega^*)$ with which Eve guesses Alice's outputs for any input string $\textbf{x}^*$ in the $\mathcal{TONS}$ scenario decreases as $P_g(\textbf{x}^*, \omega^*) \leq  e^{-\Omega(n)}$.
\end{thm}
The proof of the theorem follows by first considering a tripartite guessing game $G_g$ counterpart of the bipartite game $G$ (the idea of using guessing games in security proofs comes from \cite{V17}). Secondly, we leverage powerful $t$-out-of-$n$ parallel repetition theorems for the no-signalling value of nonlocal games by Holenstein and Buhrman et al. in \cite{Hol09, BFS14}, and consider the value of the parallel repeated guessing game. Finally, we apply suitable concentration theorems to show that Eve's guessing probability in the $\mathcal{ABNS}$ scenario (and consequently also in the $\mathcal{TONS}$ scenario) decreases exponentially in $n$.

\textit{Suitable Bell tests for $\mathcal{NS}$-secure randomness.-} We have seen that the min-entropy in the outputs increases linearly against $\mathcal{TONS}$ adversaries for inequalities such as the chain inequality. Is this property true for all Bell inequalities, which inequalities are suited for randomness or key generation against no-signalling adversaries? 
In this section, we show that in fact most bipartite pseudo-telepathy (PT) games (games with quantum value equal to no-signalling value) are \textit{unsuitable} for randomness generation against no-signalling adversaries. We prove in Appendix \ref{sec:BoundRandomness-NS} that for such games, even the observation of the maximal quantum value by Alice and Bob does not prevent a no-signalling eavesdropper from being able to guess Alice's outputs perfectly. In particular, we consider PT games with the optimal strategy involving measurements on a maximally entangled state of local dimension $d \geq 3$. This is formalised as Assumption 1 in App. \ref{sec:BoundRandomness-NS}. Note that all known PT games belong to this class and it is known that the measurements in such games define a weak Kochen-Specker set \cite{RW04}. We consider games in which the set of local contextual correlations (defining the so-called Lovasz-theta set)  are contained within a graph-theoretic set known as the (normalised) fractional stable set polytope. This is formalised as Assumption 2 in App. \ref{sec:BoundRandomness-NS}, again most of the known PT games belong to this class as shown in \cite{GKS24}. We show in the following theorem that all such PT games (which includes the known games such as the Magic Square, quantum coloring games, quantum independent-set games \cite{MSS13} etc.) are unsuitable for randomness generation against NS adversaries. To prove this, we construct an attack strategy that allows a no-signalling adversary to perfectly guess Alice's outputs even as Alice and Bob observe maximal violation of the Bell inequality.

\begin{thm}
\label{thm:PTgame-NS}
Let $G$ denote a bipartite pseudotelepathy game with respect to maximally entangled state $|\phi_d \rangle = \frac{1}{\sqrt{d}} \sum_{i=0}^{d-1} | i \rangle_A |i \rangle_B$ for some $d \geq 3$ and with the corresponding weak KS set satisfying Assumption 2. Then  for any input $x^*$ of one player, there exists a strategy that allows a no-signalling adversary to perfectly guess the outputs of $x^*$ even when the  maximum no-signaling value $\omega_{\mathcal{NS}}(G)$ is achieved, i.e., the single-round guessing probability 
\begin{equation}
P_g(x^*, \omega_{\mathcal{NS}}(G)) = 1.
\end{equation}
\end{thm}

Apart from DI protocols, it is also worth considering the implications  of Thm. \ref{thm:PTgame-NS} from a fundamental point of view. Several interesting open questions remain in quantum foundations - (i) Most non-local games display the phenomenon of monogamy of non-locality in no-signalling theories \cite{Toner09}. Specifically here we consider the tradeoff between the achievement of the maximal value of the non-local game and the correlation of a single observable involved in the game with a third player. For instance, the achievement of the no-signalling value of $1$ for the CHSH game is well-known to imply that the correlation between any observable $A_i$ of one player and $E$ of a third player is zero, i.e., when the no-signalling value of the CHSH game is achieved, the outputs become maximally random for any no-signalling adversary. Theorem \ref{thm:PTgame-NS} says that this property of monogamy does not hold for general bipartite pseudotelepathy games in no-signalling theories. In this regard, we generalise the results from \cite{RH14, ACP+16} where the property was observed for the famous Magic Square PT game. 
(ii) It is known that quantum theory does not allow the realisation of non-local extremal no-signalling behaviors \cite{RTHH16}. A fundamental question that is open is - what is the lowest dimensional face of the no-signalling set that is realisable in quantum theory? How close to extremal non-local behaviors can one get using quantum correlations? The property of reaching a nonlocal face of the no-signalling set only holds for quantum correlations achieving the maximum value in PT games. As a corollary of Theorem \ref{thm:PTgame-NS}, we observe that the minimum dimension of a nonlocal no-signalling face reached by quantum correlations in the Bell scenario $(m, d; m, d)$ (with $m$ inputs and $d \geq 3$ outputs for the two players) is lower bounded by $d-1$. Specifically, let $\mathit{F}$ denote the non-local face of the no-signalling polytope reached by the quantum behavior winning the bipartite PT game. Let $P^{\text{att}}_{\mathsf{A}\mathsf{B}|\mathsf{X}\mathsf{Y}}$ denote the set of marginal behaviors from the no-signalling attack strategies allowing Eve to perfectly guess Alice's output under the constraint that Alice-Bob achieve maximum quantum value for the inequality. Then it follows, that
\begin{equation}
\text{dim}\left(\mathit{F} \right) \geq \text{dim}\left(\textsf{aff}\left(P^{\text{att}}_{\mathsf{A}\mathsf{B}|\mathsf{X}\mathsf{Y}}\right)\right) \geq d-1,
\end{equation}
where 
\begin{equation}
\textsf{aff}\left(P^{\text{att}}_{\mathsf{A}\mathsf{B}|\mathsf{X}\mathsf{Y}}\right) := \bigcup \left\{S \; | S \; \text{is affine and} \; P^{\text{att}}_{\mathsf{A}\mathsf{B}|\mathsf{X}\mathsf{Y}} \in S \right\},
\end{equation}
where we recall that a set is affine if and only if it contains all affine combinations of its elements. The lower bound of $d-1$ comes from the fact that by Thm. \ref{thm:PTgame-NS}, 
each of the behaviors $P^{\text{att}}_{\mathsf{A}\mathsf{B}|\mathsf{X}\mathsf{Y}}$ is deterministic for one specific input-output combination $(x,a)$ of Alice. Clearly, the $d$ behaviors that are deterministic for any single $x$ are affinely independent, so that $\text{dim}\left(\textsf{aff}\left(P^{\text{att}}_{\mathsf{A}\mathsf{B}|\mathsf{X}\mathsf{Y}}\right)\right) \geq d-1$.
In this regard, we extend to dimension $d-2$ faces for such bipartite scenarios, the results of \cite{RTHH16} where it was shown that quantum correlations cannot reach the nonlocal vertices (dimension $0$ faces) of the no-signalling set.

\textit{Comparison with Quantum Adversaries.-} The chain inequality \cite{BC89} with three inputs per player has been found to be ideally suited for randomness generation against both no-signalling \cite{BHK05, BCK12} and quantum adversaries \cite{MPA11}, in both cases outperforming the commonly used CHSH test in terms of generation rate and noise tolerance. On the other hand, the analytical form of the min-entropy (and the quantum conditional von Neumann entropy) are still unknown for this inequality. Given its importance for DI protocols, we particularly investigate this inequality and analytically compare the local guessing probability for both quantum and no-signalling adversaries as a function of its violation in Appendix \ref{sec:ChainIneq-GuessProb}. 

\textit{Open Questions.-} We have seen that the raw randomness in the outputs accumulates linearly in the $\mathcal{TONS}$ scenario, it would be good to show that the smooth min-entropy is also bounded in this way. We also note that the authors of \cite{BPA18} also consider an alternative scenario that they term $\mathcal{WTONS}$ whereby future runs cannot influence past runs, and no-signalling holds at each individual run, but contrarily to $\mathcal{ABNS}$ and $\mathcal{TONS}$, no-signalling between Alice and Bob does not hold throughout the experiment. Specifically, the two devices are allowed to communicate between successive runs, so that Alice's marginal in the $k$-th run may also depend on Bob's inputs in all runs $\leq k$. The $\mathcal{WTONS}$ scenario is crucial in DI-QKD setups where one necessarily needs to allow communication between rounds of the protocol (this is of less importance in DI-QRNG setups where the protocol is executed by a single party using two shielded devices). We prove an analytical statement on the behavior of the guessing probability in the $\mathcal{WTONS}$ scenario in forthcoming work \cite{RLW25}. Finally, while we have been concerned with the accumulation of randomness in the outputs, it is also important to study the possibility of time-ordered no-signalling privacy amplification to definitively answer the question of achieving device-independent no-signalling security. In this regard, as we have mentioned while privacy amplification is possible in the $\cal{FULL-NS}$ setting, several no-go theorems have been proven in the $\mathcal{TONS}$ setting for certain classes of non-local games, it remains to be seen how far these can be extended or circumvented.



\textit{Acknowledgments.-} Useful discussions with Pawe{\l} Horodecki and Stefano Pironio are acknowledged. We acknowledge support from the General Research Fund (GRF) grant No. 17211122 and the Research Impact Fund (RIF) No. R7035-21.


\bibliographystyle{apsrev4-2}
\bibliography{common}

\onecolumngrid
\appendix



\section{Randomness in the raw output string of honest players against NS adversaries}
\label{sec:App-TONS-rand}


Consider a Bell experiment in which Alice, after interacting $n$ times with her device, ends up with input and output strings $\textbf{x} = \left(x_1, x_2, \ldots, x_n \right) \in \mathsf{X}^n := \mathsf{X} \times \ldots \times \mathsf{X}$ and $\textbf{a} = \left(a_1, a_2, \ldots, a_n \right) \in \mathsf{A}^n := \mathsf{A} \times \ldots \times \mathsf{A}$ respectively. Similarly, Bob obtains input and output strings $\textbf{y} \in \mathsf{Y}^n$ and $\textbf{b} \in \mathsf{B}^n$ respectively. The set of their conditional distributions $\{P_{\mathsf{A}^n\mathsf{B}^n|\mathsf{X}^n\mathsf{Y}^n}(\textbf{a}, \textbf{b}| \textbf{x}, \textbf{y})\}$ is termed a behavior. An adversary Eve has in her possession a device, with inputs $z \in \mathsf{Z}$ and outputs $e \in \mathsf{E}$, that may be correlated with the device held by Alice and Bob. In other words, the behavior of the three players is described by $\{P_{\mathsf{A}^n\mathsf{B}^n\mathsf{E}|\mathsf{X}^n\mathsf{Y}^n\mathsf{Z}}(\textbf{a}, \textbf{b},e|\textbf{x}, \textbf{y},z)\}$ whose marginals on Alice-Bob are given by $\{P_{\mathsf{A}^n\mathsf{B}^n|\mathsf{X}^n\mathsf{Y}^n}(\textbf{a}, \textbf{b}| \textbf{x}, \textbf{y})\}$ and which we assume to not permit signalling between Eve and Alice-Bob. The adversary aims to guess Alice's outputs after getting to know her inputs, that is, Eve aims to achieve $e = \textbf{a}$ after receiving $\textbf{x}$. 
The behavior of the device held by Alice-Bob and the conditional behaviors $P_{\mathsf{A}^n\mathsf{B}^n|\mathsf{X}^n\mathsf{Y}^n\mathsf{E}\mathsf{Z}}$ obey certain constraints depending on the causal structure being considered, here we consider the causal structures termed $\mathcal{ABNS}$ and $\mathcal{TONS}$ \cite{BPA18} defined in Def. \ref{def:ABNS-TONS}.

As we have seen, Eve's guessing probability in the $\mathcal{ABNS}$ scenario is greater than or equal to that in the $\mathcal{TONS}$ scenario. In this Appendix, we aim to show that the guessing probability against no-signaling adversaries goes down exponentially in $n$ when the marginal behavior $P_{\mathsf{A}^n\mathsf{B}^n|\mathsf{X}^n\mathsf{Y}^n}$ of Alice-Bob displays a violation for a large class of Bell inequalities, in the $\mathcal{ABNS}$ and consequently also in the $\mathcal{TONS}$ scenario.

For concreteness, we fix the Bell inequality to be the Braunstein-Caves chained Bell inequality \cite{BC89, SASA16} for $m=3$ inputs per player. The chained Bell inequality has found repeated use in proofs of security against no-signalling adversaries owing to the fact that the quantum value tends towards the no-signalling value in the limit of large $m$. In each round of the game, Alice and Bob receive respective inputs $x \in \mathsf{X} = \{0,1, 2\}$, $y \in \mathsf{Y} = \{0,1, 2\}$ with uniform probability (the input distribution is $\pi(x,y) = \frac{1}{9}$ for all $x, y$) and produce respective outputs $a \in \mathsf{A} = \{0,1\}$, $b \in \mathsf{B} = \{0,1\}$. The winning condition in the game is given by the predicate (here we consider a complete-support variant of the game)
\begin{equation}
\label{eq:pred-chaingame}
V(a,b,x,y) =   \left\{
\begin{array}{ll}
      1 & a \oplus b = 0 \wedge (x,y) \in \{(0,0),(1,0),(1,1),(2,1),(2,2)\} \\
      1 & a \oplus b = 1 \wedge (x,y) = (0,2) \\
      1 & (x,y) \in \{(0,1), (1,2), (2,0)\}\\ 
      0 & \text{otherwise} \\
\end{array} 
\right. 
\end{equation}
A strategy $\{P_{\mathsf{A}\mathsf{B}|\mathsf{X}\mathsf{Y}}(a, b| x, y)\}$ is said to belong to the set of non-signalling strategies $\mathcal{NS}$ if the output distribution of one player does not depend on the inputs of the other player. The non-signalling game value $\omega_{\mathcal{NS}}(G)$ is given as
\begin{equation}
\omega_{\mathcal{NS}}(G) = \sup_{P_{\mathsf{A}\mathsf{B}|\mathsf{X}\mathsf{Y}} \in \mathcal{NS}} \sum_{x,y,a,b} \pi(x,y) V(a,b,x,y) P_{\mathsf{A}\mathsf{B}|\mathsf{X}\mathsf{Y}}(a,b|x,y).
\end{equation}
For the $3$-input chain game, it is well-known that the no-signaling value is $1$ with the strategy involving the use of a PR-type behavior.

We are interested in the probability with which Eve can guess Alice's outputs upon coming to know Alice's inputs, under the constraint that the Bell value of the Alice-Bob behavior in each round is some specified super-classical value $\omega^*$. Specifically for any given input string $\textbf{x}^*$, Eve chooses a specific input $z^*$ to her device and attempts to obtain the output $e = \textbf{a}$. The probability with which Eve can guess Alice's output string is given by the solution to the optimisation problem given in Eq.\eqref{eq:TONS-guessprob}. We now show that the solution to this optimisation problem decreases exponentially with $n$ when the true Bell value of the marginal behavior of Alice-Bob is at least $\omega^*$ for each of the $n$ rounds. 
We emphasise that even though we prove the theorem in terms of the specific chained Bell inequality, the statement applies generally to all inequalities that exhibit the phenomenon of monogamy of nonlocality in no-signalling theories, as remarked earlier. 
\begin{thm}
\label{thm:guess-prob-tons-2}
Let $G$ denote a non-local game that exhibits the phenomenon of monogamy of non-locality in no-signalling theories.
Consider a device-independent scenario where the marginal behavior of Alice and Bob has value $\omega^*$ for a non-local game $G$ in each of $n$ rounds. Then the probability $P_g(\textbf{x}^*, \omega^*)$ with which Eve guesses Alice's outputs for any input string $\textbf{x}^*$ in the $\mathcal{TONS}$ scenario decreases as $P_g(\textbf{x}^*, \omega^*) \leq  e^{-\Omega(n)}$.
\end{thm}

\begin{proof}
Let $G$ denote the chain inequality non-local game with $m=3$ inputs per player as defined above. The three-player guessing game counterpart $G_{g}$ is defined as follows (here we consider a complete-support variant of the guessing game in \cite{V17}). In each round of the guessing game $G_{g}$, Alice, Bob and Eve receive independent and uniformly distributed inputs $x, y, z \in \{0,1, 2\}$, i.e., the input distribution $\pi(x,y,z) = \frac{1}{27}$ for all $x,y,z$. The players produce outputs $a, b, e \in \{0,1\}$ respectively. The players win if and only if the outputs of Alice and Bob satisfy the chain winning condition, while at the same time Eve's output equals that of Alice in the cases when $z = x$. For other inputs, every output wins. To be precise, the predicate $V_{g}(a,b,e,x,y,z)$ of the guessing game is defined as 


\begin{equation}
V_{g}(a,b,e,x,y,z)= \left\{ 
\begin{array}{ll}
1 & a \oplus b = 0 \wedge e = a \wedge (x,y,z) \in \{(0,0,0),(1,0,1),(1,1,1),(2,1,2),(2,2,2)\}  \\
1 & a \oplus b = 1 \wedge e = a \wedge (x,y,z) = (0,2,0)  \\
1 & a \oplus b = 0 \wedge (x,y) \in \{(0,0),(1,0),(1,1),(2,1),(2,2)\} \wedge z \neq x  \\
1 & a \oplus b = 1 \wedge (x,y) = (0,2) \wedge z \neq x  \\
1 & (x,y) \in \{(0,1), (1,2), (2,0)\} \\
0 & \text{otherwise}
\end{array}
\right.
\end{equation}
Notably, $G_{g}$ is a complete-support game, i.e., the input distribution $\pi(x,y,z) > 0$ for all triples $(x, y, z) \in \mathsf{X} \times \mathsf{Y} \times \mathsf{Z}$. 

Now, consider a single run of the guessing game in which the no-signalling constraints are relaxed by a parameter $\epsilon \geq 0$. That is, we consider the maximization of the game value by a single-round behavior $P_{\mathsf{A}\mathsf{B}\mathsf{E}|\mathsf{X}\mathsf{Y}\mathsf{Z}}$ that is $\epsilon$-almost-no-signalling. The $\epsilon$-almost-no-signalling strategy in each round is defined as follows. 
\begin{defi}
The (single-round) behavior $\{P_{\mathsf{A}\mathsf{B}\mathsf{E}|\mathsf{X}\mathsf{Y}\mathsf{Z}}(a,b,e|x,y,z)\}$ is said to be an $\epsilon$-almost-no-signalling strategy if it holds that
\begin{eqnarray}
&&\big| \sum_{a} P_{\mathsf{A}\mathsf{B}\mathsf{E}|\mathsf{X}\mathsf{Y}\mathsf{Z}}(a,b,e|x,y,z) - \sum_{a} P_{\mathsf{A}\mathsf{B}\mathsf{E}|\mathsf{X}\mathsf{Y}\mathsf{Z}}(a,b,e|x',y,z) \big| \leq \epsilon, \;\; \forall b,e,y,z, x,x', \nonumber \\
&&\big| \sum_{b} P_{\mathsf{A}\mathsf{B}\mathsf{E}|\mathsf{X}\mathsf{Y}\mathsf{Z}}(a,b,e|x,y,z) - \sum_{a} P_{\mathsf{A}\mathsf{B}\mathsf{E}|\mathsf{X}\mathsf{Y}\mathsf{Z}}(a,b,e|x,y',z) \big| \leq \epsilon, \;\; \forall a,e,x,z, y,y', \nonumber \\
&&\big| \sum_{e} P_{\mathsf{A}\mathsf{B}\mathsf{E}|\mathsf{X}\mathsf{Y}\mathsf{Z}}(a,b,e|x,y,z) - \sum_{a} P_{\mathsf{A}\mathsf{B}\mathsf{E}|\mathsf{X}\mathsf{Y}\mathsf{Z}}(a,b,e|x,y,z') \big| \leq \epsilon, \;\; \forall a,b,x,y, z,z', \nonumber \\
&&\big| \sum_{a,b} P_{\mathsf{A}\mathsf{B}\mathsf{E}|\mathsf{X}\mathsf{Y}\mathsf{Z}}(a,b,e|x,y,z) - \sum_{a,b} P_{\mathsf{A}\mathsf{B}\mathsf{E}|\mathsf{X}\mathsf{Y}\mathsf{Z}}(a,b,e|x',y',z) \big| \leq \epsilon, \;\; \forall e,z,x,y,x',y', \nonumber \\
&&\big| \sum_{a,e} P_{\mathsf{A}\mathsf{B}\mathsf{E}|\mathsf{X}\mathsf{Y}\mathsf{Z}}(a,b,e|x,y,z) - \sum_{a,e} P_{\mathsf{A}\mathsf{B}\mathsf{E}|\mathsf{X}\mathsf{Y}\mathsf{Z}}(a,b,e|x',y,z') \big| \leq \epsilon, \;\; \forall b,y,x,z,x',z', \nonumber \\
&&\big| \sum_{b,e} P_{\mathsf{A}\mathsf{B}\mathsf{E}|\mathsf{X}\mathsf{Y}\mathsf{Z}}(a,b,e|x,y,z) - \sum_{b,e} P_{\mathsf{A}\mathsf{B}\mathsf{E}|\mathsf{X}\mathsf{Y}\mathsf{Z}}(a,b,e|x,y',z') \big| \leq \epsilon, \;\; \forall a,x,y,z,y',z'.
\end{eqnarray}
\end{defi}
In the following lemma, we compute the value of the guessing game corresponding to the chain inequality achievable by single-round behaviors that are $\epsilon$-almost-no-signalling for given $\epsilon \geq 0$ . 

\begin{lem}
\label{lem:chainineq-pert}
Let $\epsilon \geq 0$ be given and let $G_g$ denote the tripartite guessing game corresponding to the chain inequality. For any $\epsilon$-almost-no-signalling behavior $P_{\mathsf{A}\mathsf{B}\mathsf{E}|\mathsf{X}\mathsf{Y}\mathsf{Z}}$ the value of $G_g$ achieved by $P_{\mathsf{A}\mathsf{B}\mathsf{E}|\mathsf{X}\mathsf{Y}\mathsf{Z}}$ is bounded as
\begin{equation}
\label{eq:chainguessgame-value}
\omega^{(P)}(G_g) \leq \left\{ 
\begin{array}{ll}
\frac{1}{9}\left(8 + 10 \epsilon \right) & \; \;  \epsilon \in [0, 1/10] \\
1 & \; \; \epsilon > 1/10
\end{array}
\right.
\end{equation} 
In particular, the no-signalling value of the guessing game $G_{g}$ is $\omega_{\mathcal{NS}}(G_{g}) = 8/9$.
\end{lem}
\begin{proof}
The $\epsilon$-almost-no-signalling value can be calculated by means of a linear program, in particular each of the absolute value constraints is expressed in terms of two linear constraints, and the above statement is proven through the strong duality of linear programming.

Specifically, the dual of the $\epsilon$-almost no-signaling LP for the guessing game $G_g$ can be written as:

\begin{equation}
    \begin{split}
        \max \qquad &\sum_{x,y,z} N_{xyz} +\epsilon\left( \sum_{\substack{x,x'\\ b,e,y,z}}\tau_{beyz}^{xx'}+\sum_{\substack{y,y'\\ a,e,x,z}}\tau_{aexz}^{yy'}+\sum_{\substack{z,z'\\ a,b,x,y}}\tau_{abxy}^{zz'}+\sum_{\substack{x,y,x'y'\\ e,z}}\tau_{ez}^{xyx'y'}+\sum_{\substack{x,z,x',z'\\ b,y}}\tau_{by}^{xzx'z'}+\sum_{\substack{y,z,y',z'\\ a,x}}\tau_{a,x}^{yzy'z'}\right)\\
        \text{s.t.} \qquad & N_{xyz}+\sum_{x'}\left(\tau_{beyz}^{xx'}-\tau_{beyz}^{x'x} \right) + \sum_{y'}\left(\tau_{aexz}^{yy'}-\tau_{aexz}^{y'y} \right) + \sum_{z'}\left(\tau_{abxy}^{zz'}-\tau_{abxy}^{z'z} \right) +\sum_{x',y'}\left(\tau_{ez}^{xyx'y'}-\tau_{ez}^{x'y'xy} \right) \\
        & +\sum_{x',z'}\left(\tau_{by}^{xzx'z'}-\tau_{by}^{x'z'xz} \right) +\sum_{y',z'}\left(\tau_{ax}^{yzy'z'}-\tau_{ax}^{y'z'yz} \right) \geq \frac{V_g(a,b,e,x,y,z)}{\pi(x,y,z)},\qquad \forall a,b,e,x,y,z,\\
        &  \tau_{beyz}^{xx'}\geq0, \qquad  \forall b,e,y,z,x,x',\\
        & \tau_{aexz}^{yy'}\geq0, \qquad \forall a,e,x,z,y,y',\\
        & \tau_{abxy}^{zz'}\geq0, \qquad \forall a,b,x,y,z,z',\\
        & \tau_{ez}^{xyx'y'}\geq0, \qquad \forall e,z,x,y,x',y',\\
        & \tau_{by}^{xzx'z'}\geq0, \qquad \forall b,y,x,z,x',z',\\
        & \tau_{ax}^{yzy'z'}\geq0, \qquad \forall a,x,y,z,y',z'.\\
    \end{split}
\end{equation}
where $N_{xyz}$ is the dual variable associated with the normalization condition for the setting $x,y,z$, and the variables $\tau$s are the dual variables associated with the $\epsilon$-almost no-signaling conditions. Solving the above linear program and using strong duality gives Eq.\eqref{eq:chainguessgame-value}. 
The no-signalling value of $G_g$ can be obtained from the solution by setting $\epsilon = 0$.
\end{proof}
While we have computed the $\epsilon$-almost-no-signalling value for the chain game, in \cite{BFS14, Hol09} it was proven, based on the sensitivity analysis of linear programs from \cite{Schrijver98}, that for for all games the relaxation of the constraints in the linear program for computation of no-signalling value only perturbs the value by $\alpha(G) \cdot \epsilon$ for $\alpha(G)$ depending upon the number of inputs and outputs in the game. To elaborate, the following lemma shows how the no-signaling game value changes under epsilonic amounts of signalling.
\begin{lem}[\cite{BFS14, Hol09}]
\label{lem:almost-no-signaling}
Let $G$ be a $3$-player game with no-signalling value $\omega_{\mathcal{NS}}(G)$. Then there exists a constant $\alpha(G)$ depending only upon the number of inputs and outputs of $G$, such that for any $\epsilon > 0$, and for any $\epsilon$-almost-no-signalling behavior $\{P_{\mathsf{A}\mathsf{B}\mathsf{E}|\mathsf{X}\mathsf{Y}\mathsf{Z}}(a,b,e|x,y,z)\}$, the value of $G$ achieved by $\{P_{\mathsf{A}\mathsf{B}\mathsf{E}|\mathsf{X}\mathsf{Y}\mathsf{Z}}(a,b,e|x,y,z)\}$ is bounded as
\begin{equation}
\label{eq:almost-no-signaling}
\omega^{\{P\}}(G) \leq \omega_{\mathcal{NS}}(G) + \alpha(G) \cdot \epsilon.
\end{equation}
\end{lem}

Now, consider the $t$-out-of-$n$ Parallel Repetition of the game $G_{g}$ denoted as $G_{g}^{t/n}$ \cite{BFS14}. In the game $G_{g}^{t/n}$, the players receive inputs $\textbf{x}, \textbf{y}, \textbf{z}$ respectively where $\textbf{x} = \left(x_1, \ldots, x_n \right)$ and similarly for the other inputs. The inputs are independent between all the $n$ rounds, so that the input distribution $\pi^n$ in the game is given by $\pi^n(\textbf{x}, \textbf{y}, \textbf{z}) = \pi(x,y,z)^n$. The players produce outputs $\textbf{a}, \textbf{b}, \textbf{e}$ respectively where again $\textbf{a}^n = \left(a_1, \ldots, a_n \right)$ and similarly for the other outputs. The winning condition in the game is given by the predicate $V^{t/n}_{g}(\textbf{a}, \textbf{b}, \textbf{e}, \textbf{x}, \textbf{y}, \textbf{z})$ defined as
\begin{equation}
V^{t/n}_{g}(\textbf{a}, \textbf{b}, \textbf{e}, \textbf{x}, \textbf{y}, \textbf{z})= \left\{ 
\begin{array}{ll}
1 & \; \; \sum_{i=i}^n V_{g}(a_i, b_i, e_i, x_i, y_i, z_i) \geq t\\
0 & \; \; \text{otherwise}
\end{array}
\right.
\end{equation}
The usual parallel repetition considers the scenario when $t=n$. In \cite{BFS14}, the following concentration theorem was proven.

\begin{thm}[\cite{BFS14}]
\label{thm-t-out-of-n-parallelrep}
Let $G$ be an arbitrary $m$-player game with complete support. Then there exists a constant $\mu > 0$ depending on $G$ such that for any $\delta > 0$, any $n \in \mathbb{N}$ and for $t = \left( \omega_{\mathcal{NS}}(G) + \delta \right) n$, it holds that
\begin{equation}
\omega_{\mathcal{NS}}(G^{t/n}) \leq 8 \exp\left(- \delta^4 \mu n \right).
\end{equation}
\end{thm}

Applying Theorem \ref{thm-t-out-of-n-parallelrep} to the guessing game $G_{g}$ corresponding to the chain expression we have that $\omega_{\mathcal{NS}}(G_{g}^{t/n}) \leq 8 \exp\left(- \delta^4 \mu n \right)$. 
Equivalently that for any behavior $P_{\mathsf{A}^n\mathsf{B}^n\mathsf{E}^n|\mathsf{X}^n\mathsf{Y}^n\mathsf{Z}^n}$ that is non-signalling of the $\mathcal{ABNS}$ type (such that the outputs of any subset of players do not depend upon the inputs of the complementary set but that the outputs of one player can depend upon any of their inputs), and for any subset $T \subset [n]$ of size $|T| \geq t$ it holds that
\begin{eqnarray}
\label{eq:guess-conc}
\text{Pr}\left[ \bigwedge_{i \in T} V_{g}(a_i,b_i,e_i, x_i, y_i, z_i) = 1 \right] \leq 8 e^{- \delta^4 \mu n}.
\end{eqnarray}
Specifically for the $m=3$ chain inequality under consideration, we can compute the constants using Lemma \ref{lem:chainineq-pert} and following the proof of Thm. \ref{thm-t-out-of-n-parallelrep} from \cite{BFS14} as
\begin{eqnarray}
\mu = \frac{\pi_{\min}^2}{\alpha(G_g)^2 \cdot 6^7} = \frac{1}{6^9 \cdot 5^2}, \quad \delta < \frac{1}{10},
\end{eqnarray}
where we have used that the minimum probability of the inputs in the complete-support guessing game $G_g$ is $\pi_{\min} = \frac{1}{3^3}$.
Let us define a predicate $W(a,e,x,z)$ as
\begin{equation}
W(a,e,x,z)= \left\{ 
\begin{array}{ll}
1 & \; \; a =e \wedge x = z\\
1 & \; \; x \neq z \\
0 & \; \; a \neq e \wedge x = z
\end{array}
\right.
\end{equation}
Then by definition we have that
\begin{equation}
\label{eq:guess-conjunction}
V_g(a_i, b_i, e_i, x_i, y_i, z_i) = V(a_i, b_i, x_i, y_i) \cdot W(a_i, e_i, x_i, z_i). 
\end{equation}
Now consider that in a device-independent protocol, Alice and Bob test for the violation of the chain inequality in $n$ runs of the protocol. Specifically, they test that their device produces winning outputs in $t$ out of $n$ runs, and abort the protocol if the test fails. With the given condition that the marginal behavior of Alice and Bob has true Bell value $\omega^*$ in each of the $n$ rounds, a standard concentration argument via the Azuma-Hoeffding inequality gives that the fraction of rounds in which Alice and Bob observe winning outcomes for the chain inequality is $\geq \omega^* - \kappa$ with probability at least $1 - 2 \exp\left(- n \kappa^2/2 \right)$.  


That is, defining the event $\mathtt{ABORT}$ as 
\begin{equation}
\mathtt{ABORT} := \bigg\{ (\textbf{a},\textbf{b}, \textbf{x}, \textbf{y}) \bigg| \sum_{i=1}^n V(a_i,b_i, x_i, y_i) < t \bigg\},
\end{equation}
and under the above assumption on the marginal behavior of Alice-Bob we have by  the Azuma-Hoeffding inequality that for $t = (\omega^* - \kappa)n$,  
\begin{eqnarray}
\label{eq:prob-abort}
\text{Pr}\left[\mathtt{ABORT} \right] \leq 2 e^{-n \kappa^2/ 2}.
\end{eqnarray}
Here, we note that since $t = (\omega^* - \kappa)n = (\omega_{\mathcal{NS}}(G_g) + \delta)n$ from Thm. \ref{thm-t-out-of-n-parallelrep}, and with $\omega_{\mathcal{NS}}(G_g) = 8/9 \approx 0.889$ from Lemma \ref{lem:chainineq-pert} and the maximum quantum value of the chain game with predicate \eqref{eq:pred-chaingame} $\omega_{\mathcal{NS}}(G) = (4+\sqrt{3})/6 \approx 0.955$ from the strategy given by the following measurements \cite{SASA16} on the maximally entangled state $|\phi_2 \rangle = \frac{1}{\sqrt{2}} \left(|0\rangle_A |0 \rangle_B + |1 \rangle_A |1 \rangle_B \right)$
\begin{eqnarray}
A_x = \sin(\theta_x) \sigma_x + \cos(\theta_x) \sigma_z, \nonumber \\
B_y = \sin(\phi_y) \sigma_x + \cos(\phi_y) \sigma_z,
\end{eqnarray} 
with $\theta_x =  x \cdot \pi/3$ and $\phi_y = (2y+1) \cdot\pi/6$ for $x,y = 0,1,2$, we can set
\begin{equation}
\omega^* = \kappa + \delta + \frac{8}{9} \leq \frac{4+\sqrt{3}}{6},
\end{equation}
for suitably small choice of $\kappa > 0$ in the Azuma-Hoeffding inequality, and  $\delta > 0$ in Thm. \ref{thm-t-out-of-n-parallelrep}. 

Thus far we have, under the constraint that the marginal behaviors achieve value $\omega^*$ in each of the $n$ rounds, that for some subset $T \subset [n]$ of size $|T| \geq t$ with probability at least $1 - 2 \exp\left(-n \kappa^2/ 2\right)$, it holds that $\bigwedge_{i \in T} V(a_i,b_i, x_i, y_i) = 1$. Equivalently with $\overline{\mathtt{ABORT}}$ denoting the event that Alice and Bob do not abort the protocol, we have that $\text{Pr}\left[\overline{\mathtt{ABORT}} \right] \geq 1 - 2 \exp\left(-n \kappa^2/ 2\right)$.

Combining with Eq.\eqref{eq:guess-conjunction}, we obtain that
\begin{equation}
\label{eq:guess-given-abort}
\text{Pr}\left[ \bigwedge_{i \in T} W(a_i,e_i, x_i, z_i) = 1 \;  \bigg| \;  \overline{\mathtt{ABORT}} \right] \leq  \frac{8 e^{- \delta^4 \mu n}}{\text{Pr}\left[\overline{\mathtt{ABORT}}\right]},
\end{equation}
with $\text{Pr}\left[\overline{\mathtt{ABORT}} \right] \geq 1 - 2 \exp\left(-n \kappa^2/ 2\right)$.



Now, define $\mathtt{Guess-Input}$ as the particular subset of $T$ (the set of winning runs) in which Eve chooses an input $z_i$ equal to Alice's input $x_i$. That is 
\begin{equation}
\mathtt{Guess-Input} := \big\{ i \in T \; | \; z_i = x_i \big\}. 
\end{equation}
Define a random variable $H_i$ taking value $1$ when $z_i \neq x_i$ and $0$ otherwise. Since the inputs are chosen independently in each round so that the probability of $z_i \neq x_i$ is $2/3$, we have by the Generalised Chernoff bound \cite{PS97} that
\begin{equation}
\text{Pr}\left[ \sum_{i=1}^n H_i \geq \gamma t \right] \leq e^{- 2 t \left(\gamma - 2/3 \right)^2 },
\end{equation}
for $2/3 < \gamma < 1$ or equivalently $\text{Pr}\left[ \sum_{i=1}^n H_i < \gamma t \right] \geq 1 - e^{- 2 t \left(\gamma - 2/3 \right)^2 }$. 
Here, we have used the following Generalised Chernoff bound Eq.\eqref{eq:gen-Chernoff} from \cite{PS97}.
\begin{thm}[\cite{PS97}]
Let $X_i$ for $i \in \{1,\ldots, t\}$ be Boolean random variables such that for some $0 \leq \zeta \leq 1$ we have that for every subset $S \subseteq \{1,\ldots, t\}$, $\text{Pr}\left[\wedge_{i \in S} X_i = 1 \right] \leq \zeta^{|S|}$. Then for any $0 \leq \zeta \leq \gamma \leq 1$ it holds that
\begin{equation}
\label{eq:gen-Chernoff}
\text{Pr}\left[\sum_{i=1}^t X_i \geq \gamma t \right] \leq e^{-t D(\gamma \| \zeta)},
\end{equation}
where $D(\cdot \| \cdot)$ denotes the relative entropy function. In particular, $D(\gamma \| \zeta) \geq 2\left(\gamma - \zeta \right)^2$.
\end{thm}

In other words, we have that with probability at least $1 - e^{- 2 t \left(\gamma - 2/3 \right)^2 }$, the number of runs in which $z_i = x_i$ (the size of $\mathtt{Guess-Input}$) is at least $(1- \gamma)t$. 
We thus obtain
\begin{equation}
\label{eq:guess-given-abort-2}
\text{Pr}\left[ \bigwedge_{i \in \mathtt{Guess-Input}} \left(e_i = a_i \right) \;  \bigg| \;  \overline{\mathtt{ABORT}} \right] \leq  \frac{8 e^{- \delta^4 \mu n}}{\text{Pr}\left[\overline{\mathtt{ABORT}}\right]},
\end{equation}
with $\text{Pr}\left[\overline{\mathtt{ABORT}} \right] \geq 1 - 2 e^{-n \kappa^2/ 2}$ and where $|\mathtt{Guess-Input}| \geq (1 - \gamma) t$ with probability at least $1 - e^{- 2 t \left(\gamma - 2/3 \right)^2 }$.

Now, we observe that the probability $P_g(\textbf{x}^*, \omega^*)$ that Eve guesses Alice's output in all $n$ runs is upper bounded by the quantity on the left in Eq.\eqref{eq:guess-given-abort-2} in the causal scenario $\mathcal{ABNS}$. Specifically in this causal scenario, Eve inputs a single $z^*$ into her device and attempts to output $e = \textbf{a}$ as seen in the guessing probability $P_g(\textbf{x}^*, \omega^*)$ definition in Eq.\eqref{eq:TONS-guessprob}. We can parse the output $e$ as the string $(e_1, e_2, \ldots, e_n)$ where Eve hopes to achieve $e_i = a_i$ for all $i = 1, \ldots, n$. Therefore, the probability 
$P_g(\textbf{x}^*, \omega^*)$ under the condition that the marginal behavior of Alice and Bob's device has value $\omega^*$ in Eq.\eqref{eq:TONS-guessprob} in the $\mathcal{ABNS}$ scenario  and consequently also in the $\mathcal{TONS}$ setting is upper bounded by the right hand side of Eq.\eqref{eq:guess-given-abort-2}. We thus conclude that the guessing probability in the $\mathcal{TONS}$ scenario goes down exponentially as $P_g(\textbf{x}^*, \omega^*) \leq 24 e^{-\delta^4 \mu n}$.

\end{proof}


\section{No-Signaling strategies allowing Eve to Perfectly Guess Local Outcomes in Bipartite Pseudotelepathy games}
\label{sec:BoundRandomness-NS}

As we have seen in the previous section, certain inequalities such as the chain inequality exhibit a phenomenon of monogamy of non-locality in no-signalling theories, which implies that a no-signalling Eve is unable to perfectly guess Alice's outcomes when Alice and Bob observe maximum violation of the chain inequality. This phenomenon is thus characterised by a gap between the no-signalling value $\omega_{\mathcal{NS}}(G)$ of the bipartite game and the no-signalling value of its tripartite guessing game counterpart which we consider in the proof of Theorem \ref{thm:guess-prob-tons-2} (the notion of a guessing game was introduced in \cite{V17}). This gap is the non-locality analog of the monogamy of entanglement which has been used as the basis of security proofs in the quantum setting, see for example \cite{CMT25, V17, JMS20, JK21}. For games which display this monogamy property, the exponential decrease of guessing probability with $n$ in a $\mathcal{TONS}$ setting was proven in the previous section. 

In this Appendix, we show that this monogamy-of-nonlocality does not hold for most bipartite non-local pseudotelepathy games. As such, unlike in the case of the quantum adversary \cite{ZMZ+23}, these games are unsuitable and should not be used as the basis for device-independent randomness certification or key generation against no-signalling adversaries. 

Pseudo-telepathy (PT) games \cite{BCT99} are distributed tasks that can be perfectly achieved with shared quantum - but not classical - information. Specifically, two distant parties who do not communicate but are allowed to share a certain entangled quantum state can satisfy a deterministic condition on their mutual input-output behavior with certainty, where parties without shared entanglement cannot do so. PT games are interesting from a foundational point of view as a simple `all-versus-nothing' proof of quantum non-locality. In recent years, these games have also found applications as necessary resources in protocols of device-independent randomness amplification \cite{CR12, our16}, in demonstrating quantum computational advantage for shallow circuits \cite{BGK18} and in the proof of MIP* = RE \cite{JNVWY21}. A connection between KS vector sets and PT games has been found by Renner and Wolf in \cite{RW04}. Specifically, every so-called 'weak' KS set leads to a PT game and every PT game in which the optimal quantum strategy uses a maximally entangled state leads to a weak KS set. Formally, PT games and weak KS sets are defined as follows.
\begin{defi}[\cite{RW04}]
Let $|\psi \rangle \in \mathcal{H}_1 \otimes \mathcal{H}_2$ be a pure state. A pseudo-telepathy (PT) game with respect to $|\psi \rangle$ is a pair $(\mathsf{X}, \mathsf{Y})$ where $\mathsf{X}$ ($\mathsf{Y}$) is a set of orthonormal bases with respect to $\mathcal{H}_1$ ($\mathcal{H}_2$) such that the following holds. Let $\bar{V}$ be the following function defined on $\mathsf{X} \times \mathsf{Y}$: $\bar{V}(\mathsf{x}, \mathsf{y})$ is the set of pairs $(u,w) \in \mathsf{x} \times \mathsf{y}$ satisfying 
\begin{equation}
\langle \psi | u, w \rangle \neq 0,
\end{equation}
i.e., the outcomes $(u,w)$ have non-zero probability of occurring when $|\psi \rangle$ is measured in the basis $\mathsf{x} \times \mathsf{y}$. Then we must have that for every pair of functions $(f_1, f_2)$ (a classical strategy) where $f_1$ ($f_2$) is defined on $\mathsf{X}$ ($\mathsf{Y}$) and $f_1(\mathsf{x}) \in \mathsf{x}$, $f_2(\mathsf{y}) \in \mathsf{y}$ holds for every $\mathsf{x}, \mathsf{y}$ there must exist at least one basis $\mathsf{x}, \mathsf{y}$ such that 
\begin{equation}
(f_1(\mathsf{x}), f_2(\mathsf{y})) \notin \bar{V}(\mathsf{x}, \mathsf{y})
\end{equation}
holds. 
\end{defi}
\begin{defi}[\cite{RW04}]
A weak Kochen Specker set in $\mathbb{C}^d$ is a set $S \subset \mathbb{C}^d$ of unit vectors such that for every function $f: S \rightarrow \{0,1\}$ satisfying that for every orthonormal basis $\mathsf{x} \subset S$ of $\mathbb{C}^d$ 
\begin{equation}
\sum_{v \in \mathsf{x}} f(v) = 1
\end{equation}
holds, there exist vectors $(u,w) \in S$ with $\langle u|w \rangle = 0$ and $f(u) = f(w) = 1$.
\end{defi}

The paradigmatic example of a PT game is the famous Magic Square game $G_{MS}$. An interesting property of this game was found in \cite{ACP+16, RH14} with respect to device-independent randomness certification against no-signalling adversaries, namely that even achieving the maximum value (of $1$) for this game is not sufficient to certify device-independent randomness against no-signalling adversaries. This phenomenon of \textit{bound randomness} therefore indicates that the monogamy of non-locality in some cases does not hold in no-signalling theories. That is, unlike the case of the chain inequality explored in the proof of Theorem \ref{thm:guess-prob-tons-2}, the violation of the Bell Inequality corresponding to the Magic Square game cannot be used in the tasks of DI randomness expansion or amplification against no-signalling adversaries. In this section, we further explore this question of randomness certification against no-signalling adversaries. We prove that the phenomenon of bound randomness against NS adversaries is not unique to the Magic Square game, and in fact is true for most PT games which satisfy a couple of assumptions that we detail below.

\begin{itemize}
\item \textit{Assumption 1.-} The first assumption in this section is that the PT game is with respect to the maximally entangled state $|\phi_d \rangle = \frac{1}{\sqrt{d}} \sum_{i=0}^{d-1} |i \rangle_A |i \rangle_B$ for some $d \geq 3$. 

This is not a very strict assumption, since all known bipartite PT games are won by maximally entangled states of local dimension $d \geq 3$. In fact, it is believed that the maximally entangled state is necessary to achieve quantum pseudo-telepathy and it is a well-known open problem to identify a PT game that is won by a non-maximally entangled state \cite{Mancinska14}. 

\item \textit{Assumption 2.-} Let $S$ denote the weak KS set (with orthogonality graph $\Gamma$) corresponding to the PT game. The second assumption is that the set of local quantum contextual correlations for S, which is given by the Lovasz-theta set $TH(\Gamma)$ is contained within the intersection of the fractional stable set polytope of $\Gamma$ and the hyperplanes corresponding to the normalisation conditions for each maximal clique in $\Gamma$. 

This is a more technical assumption that is nevertheless true of known KS sets. The idea is the following - each KS set can be represented by its orthogonality graph where each vertex $v$ represents a vector $|v\rangle$ and two vertices are connected by an edge $v \sim w$ if the corresponding vectors are orthogonal $\langle v | w \rangle = 0$. The set of contextual quantum correlations corresponds to the Lovasz-theta set of this graph denoted as $TH(\Gamma)$. Formally, for a graph $\Gamma = (V_{\Gamma}, E_{\Gamma})$, the fractional stable set polytope $FSTAB(\Gamma)$ is defined as\cite{OP23}:
\[\text{FSTAB}(\Gamma) = \left\{ \overline{x} \in \mathbb{R}^{|V_{\Gamma}|} \,\middle|\, x_i+x_j \leq 1 \text{ for all } (i,j) \in E_{\Gamma} \right\}.
\]
It is well-known in graph theory that the Lovasz-theta set is contained within this larger set $FSTAB(\Gamma)$.
Our second assumption is that this holds true even when we consider the subset of points within $FSTAB(\Gamma)$ that obey the normalisation constraint that the sum of probabilities in every basis is $1$. 
\end{itemize}

We now state and prove the main theorem of this Appendix.
\begin{thm}
\label{thm:PT-guessinggame}
Let $G$ denote a bipartite pseudotelepathy game with respect to maximally entangled state $|\phi_d \rangle = \frac{1}{\sqrt{d}} \sum_{i=0}^{d-1} | i \rangle_A |i \rangle_B$ for some $d \geq 3$ and with the corresponding weak KS set satisfying Assumption 2. Then  for any input $x^*$ of one player, there exists a strategy that allows a no-signalling adversary to perfectly guess the outputs of $x^*$ even upon observing maximum violation, i.e., the single-round guessing probability 
\begin{equation}
P_g(x^*, \omega_{\mathcal{NS}}(G)) = 1.
\end{equation}
\end{thm}

\begin{proof}
Consider the bipartite PT game $G$ with respect to $|\phi_d \rangle$ consisting of sets of orthonormal bases $(\mathsf{X}, \mathsf{Y})$ of $\mathbb{C}^d$. By Theorem 5 in \cite{RW04}, we have that 
\begin{equation}
S = \bigcup_{\mathsf{x} \in \mathsf{X}} \mathsf{x} \cup \bigcup_{\mathsf{y} \in \mathsf{Y}} \mathsf{y}
\end{equation} 
is a weak KS set in $\mathbb{C}^d$. Furthermore, the predicate $V$ defining the game is characterised by the set of zero probability events. That is, for each $(\mathsf{x}, \mathsf{y}) \in \mathsf{X} \times \mathsf{Y}$ and for every pair $|u \rangle \in \mathsf{x}, |w \rangle \in \mathsf{y}$ such that $\langle u|w \rangle = 0$ we have that $V(\mathsf{u}, \mathsf{w}, \mathsf{x}, \mathsf{y}) = 0$ and is $1$ otherwise. We first consider an augmented version $\tilde{G}$ of the game in which both players measure all the bases in the weak KS set $S$ with the winning condition again defined by the set of non-zero probability events. More precisely, the input alphabets in $\tilde{G}$ are given as $\tilde{\mathsf{X}} = \tilde{\mathsf{Y}} = S$ and $V(\mathsf{u}, \mathsf{w}, \mathsf{x}, \mathsf{y}) = 0$ for every pair $\mathsf{u}, \mathsf{w} \in S$ satisfying $\langle u|w \rangle = 0$. We will prove that an attack strategy corresponding to a no-signalling behavior exists allowing Eve to guess Alice's outcomes in the game $\tilde{G}$ while at the same Alice-Bob satisfy all the winning conditions in $\tilde{G}$, which will automatically imply the same for the sub-game $G$. 

We first invoke a foundational result from polyhedral combinatorics that plays a crucial role in our argument. This result characterizes the extreme points of the fractional stable set polytope and ensures all extremal points have probability values in the set $\{0, \frac{1}{2}, 1\}$. It ensures that any vertex (i.e., extremal point) of the fractional stable set polytope is a $\{0,\frac{1}{2},1\}$-valued.

\begin{thm}[Balinski\cite{BA65}, Nemhauser and Trotter\cite{NT75}]
    \label{thm:balinski}
    Let $\overline{x} \in \mathbb{R}^{V_{\Gamma}}$ be a vertex of $\operatorname{FSTAB}(\Gamma)$. Then for every vertex $v \in V_{\Gamma}$, it holds that $x_v \in \{0,\frac{1}{2},1\}$, meaning that every vertex is $\{0,\frac{1}{2},1\}$-value assigned.
\end{thm}

This result plays a central role in our construction of no-signaling strategies. In particular, the $\{0, \frac{1}{2}, 1\}$-valued assignments described above also characterise the extreme points of the intersection of $FSTAB(\Gamma)$ with the normalisation hyperplanes (since these are supporting hyperplanes of $FSTAB(\Gamma)$). We use this result to first show the existence of a local consistent behavior with entries in $\{0, \frac{1}{2}, 1\}$ assigning value $1$ to any particular projector measured by Alice.

\begin{lem}
\label{lem:localf-strat}
Let $S \subset \mathbb{C}^d$ be a Kochen–Specker (KS) set satisfying Assumption 2.
Then, for any $|v\> \in S$, there exists a valid assignment $f : S \to \{0, \frac{1}{2}, 1\}$ such that $f(|v\>) = 1$.
\end{lem}

\begin{proof}
Let $S$ be a KS set in $\mathbb{C}^d$, where $d\geq 3$, and let $|v\> \in S$ be arbitrary. The quantum assignment $f_q$ of a KS set $S \in \mathbb{C}^d$ with $d \geq 3$ is defined for quantum state $\rho$ as $f_q(|w\>):=\Tr(\rho |w\>\<w|)$. For any $|v\> \in S$, considering the state $\rho = |v\>\<v|$ results in an assignment $f_q$ with $f_q(|v\>)= \langle \psi|P_v| \psi \rangle = |\langle v|v \rangle|^2 =1$.


By Thm. \ref{thm:balinski}, every vertex of FSTAB is $\{0, \tfrac{1}{2}, 1\}$-valued, and any point in FSTAB is a convex combination of such extremal assignments. Since $f_q(|v\rangle) = 1$ for the quantum assignment which by Assumption 2 lies within the intersection of FSTAB and the normalisation hyperplanes, and this value is preserved under convex combinations only when all contributing vertices satisfy $f(|v\rangle) = 1$, there must exist at least one extremal assignment $f: S \rightarrow \{0,\frac{1}{2},1\}$ such that $f(|v\rangle)=1$.

\end{proof}

In the following, we construct a no-signaling behavior for the Alice-Bob system, i.e., $\{P_{\mathsf{U}\mathsf{W}|\mathsf{X}\mathsf{Y}}(\mathsf{u}, \mathsf{w} \mid \mathsf{x}, \mathsf{y})\}$ that achieves a winning probability of $1$ for the PT game $\tilde{G}$, from any valid assignment $f: S \rightarrow \{0, \frac{1}{2}, 1\}$.

\begin{lem}
\label{lem:bipNS-from-localf}
Let $f: S \to \{0, \frac{1}{2}, 1\}$ be a valid assignment on a KS set $S \subset \mathbb{C}^d$. Then, there exists a bipartite no-signaling behavior $\{P^{(f)}_{\mathsf{U}\mathsf{W}|\mathsf{X}\mathsf{Y}}(\mathsf{u}, \mathsf{w} \mid \mathsf{x}, \mathsf{y})\}$ with marginals $P^{(f)}_{\mathsf{W}|\mathsf{Y}}(\mathsf{w}|\mathsf{y}) = f(|w\rangle)$ and $P^{(f)}_{\mathsf{U}|\mathsf{X}}(\mathsf{u}|\mathsf{x}) = f(|u\rangle)$. Furthermore, for any pair $|u\>, |w\> \in S$ with $\langle u | w \rangle = 0$, we have $P^{(f)}_{\mathsf{U}\mathsf{W}|\mathsf{X}\mathsf{Y}}(\mathsf{u}, \mathsf{w} \mid \mathsf{x}, \mathsf{y}) = 0$ for measurement bases $\mathsf{x}, \mathsf{y}$ with $\mathsf{u} \in \mathsf{x}$ and $\mathsf{w} \in \mathsf{y}$.
\end{lem}


\begin{proof}
Construct the bipartite behavior $\{P^{(f)}_{\mathsf{U}\mathsf{W}|\mathsf{X}\mathsf{Y}}(\mathsf{u}, \mathsf{w} \mid \mathsf{x}, \mathsf{y})\}$ as 
\begin{equation}
    P^{(f)}_{\mathsf{U}\mathsf{W}|\mathsf{X}\mathsf{Y}}(\mathsf{u}, \mathsf{w} \mid \mathsf{x}, \mathsf{y}) = P^{(f)}_{\mathsf{U}|\mathsf{W}\mathsf{X}\mathsf{Y}}(\mathsf{u} \mid \mathsf{w}, \mathsf{x}, \mathsf{y}) \cdot P^{(f)}_{\mathsf{W}|\mathsf{X}\mathsf{Y}}(\mathsf{w} \mid \mathsf{x}, \mathsf{y}).
\end{equation}
with the terms being defined as 
\begin{equation}
P^{(f)}_{\mathsf{W}|\mathsf{X}\mathsf{Y}}(\mathsf{w} \mid \mathsf{x}, \mathsf{y}) := f(|w\>) \; \; \; \; \forall \mathsf{w}, \mathsf{x}, \mathsf{y},
\end{equation}
and
\begin{equation}
P^{(f)}_{\mathsf{U}|\mathsf{W}\mathsf{X}\mathsf{Y}}(\mathsf{u} \mid \mathsf{w}, \mathsf{x}, \mathsf{y}) = \left\{ 
\begin{array}{ll}
0 & \; \; \langle u | w \rangle = 0\\
1 & \; \; \text{otherwise} \\
\end{array}
\right. \; \;  \; \forall \mathsf{x}, \mathsf{y} \; \text{with} \; \mathsf{u} \in \mathsf{x} \; \text{and} \; \mathsf{w} \in \mathsf{y}.
\end{equation}
Note that for a given $\mathsf{x}, \mathsf{y}$ and $\mathsf{w} \in \mathsf{y}$, it is impossible that for all $\mathsf{u} \in \mathsf{x}$ the corresponding vectors $|u\>$ are all orthogonal to $|w\>$ in a Hilbert space of dimension $d$. Therefore, there must exist some $\mathsf{u'} \in \mathsf{x}$ such that $P^{(f)}_{\mathsf{U}|\mathsf{W}\mathsf{X}\mathsf{Y}}(\mathsf{u'} \mid \mathsf{w}, \mathsf{x}, \mathsf{y}) = 1$, ensuring that the normalization condition is satisfied.

Using the above definitions, we see that
\begin{equation}
    P^{(f)}_{\mathsf{U}\mathsf{W}|\mathsf{X}\mathsf{Y}}(\mathsf{u}, \mathsf{w} \mid \mathsf{x}, \mathsf{y}) = P^{(f)}_{\mathsf{U}|\mathsf{W}\mathsf{X}\mathsf{Y}}(\mathsf{u} \mid \mathsf{w}, \mathsf{x}, \mathsf{y}) \cdot P^{(f)}_{\mathsf{W}|\mathsf{X}\mathsf{Y}}(\mathsf{w} \mid \mathsf{x}, \mathsf{y}) = \begin{cases}
        f(\mathsf{w}),\; &\text{when } P_{\mathsf{U}|\mathsf{W}\mathsf{X}\mathsf{Y}}(\mathsf{u} \mid \mathsf{w}, \mathsf{x}, \mathsf{y}) = 1; \\
        0,\; &\text{when } P_{\mathsf{U}|\mathsf{W}\mathsf{X}\mathsf{Y}}(\mathsf{u} \mid \mathsf{w}, \mathsf{x}, \mathsf{y}) = 0.
    \end{cases}
\end{equation}
By construction for any pair $
|u \rangle, |w \rangle \in S$, we have $P^{(f)}_{\mathsf{U}\mathsf{W}|\mathsf{X}\mathsf{Y}}(\mathsf{u}, \mathsf{w} \mid \mathsf{x}, \mathsf{y}) = 0$ for measurement bases $\mathsf{x}, \mathsf{y}$ with $\mathsf{u} \in \mathsf{x}$ and $\mathsf{w} \in \mathsf{y}$. 

Furthermore, the bipartite behavior is no-signalling. By construction, the marginals $P^{(f)}_{\mathsf{W}|\mathsf{Y}}(\mathsf{w} \mid \mathsf{y})$ are well-defined.
For the marginals $\{P^{(f)}_{\mathsf{U}|\mathsf{X}}(\mathsf{u} \mid \mathsf{x})\}$, first consider the case when $\mathsf{y} = \mathsf{x}$:
\begin{equation}\label{marg_pro_x}
\begin{split}
\sum_{\mathsf{w}} P^{(f)}_{\mathsf{U}\mathsf{W}|\mathsf{X}\mathsf{Y}}(\mathsf{u}, \mathsf{w} \mid \mathsf{x}, \mathsf{y} = \mathsf{x}) &= \sum_{\mathsf{w} : \mathsf{w} \neq \mathsf{u}} P^{(f)}_{\mathsf{U}\mathsf{W}|\mathsf{X}\mathsf{Y}}(\mathsf{u}, \mathsf{w} \mid \mathsf{x}, \mathsf{x}) + \sum_{\mathsf{w} : \mathsf{w} = \mathsf{u}} P^{(f)}_{\mathsf{U}\mathsf{W}|\mathsf{X}\mathsf{Y}}(\mathsf{u}, \mathsf{w} \mid \mathsf{x}, \mathsf{x}) \\
&= \sum_{\mathsf{w} : \<u|w\> = 0} P^{(f)}_{\mathsf{U}\mathsf{W}|\mathsf{X}\mathsf{Y}}(\mathsf{u}, \mathsf{w} \mid \mathsf{x}, \mathsf{x}) + \sum_{\mathsf{w} : \<u|w\> \neq 0} P^{(f)}_{\mathsf{U}\mathsf{W}|\mathsf{X}\mathsf{Y}}(\mathsf{u}, \mathsf{w} \mid \mathsf{x}, \mathsf{x}) \\
&= f(|u\>).
\end{split}
\end{equation}

Furthermore, when defining the conditional probabilities $P^{(f)}_{\mathsf{U}|\mathsf{W}\mathsf{X}\mathsf{Y}}(\mathsf{u} \mid \mathsf{w}, \mathsf{x}, \mathsf{y})$, we use only the local assignment $f$ and the orthogonality relations between vectors in $S$. In other words, we have
\begin{equation}
        P^{(f)}_{\mathsf{U}|\mathsf{W}\mathsf{X}\mathsf{Y}}(\mathsf{u} \mid \mathsf{w}, \mathsf{x}, \mathsf{y}) = P^{(f)}_{\mathsf{U}|\mathsf{W}\mathsf{X}}(\mathsf{u} \mid \mathsf{w}, \mathsf{x})
\end{equation}
So that the marginal probability in Eq.~\eqref{marg_pro_x} holds for any $\mathsf{y} \subset S$. Therefore,
\begin{equation}
\begin{split}
    P^{(f)}_{\mathsf{U}|\mathsf{X}}(\mathsf{u} \mid \mathsf{x}) = f(|u\>), \;\forall \mathsf{u}, \mathsf{x}.
\end{split}
\end{equation}
We have thus constructed a bipartite no-signalling behavior that wins the PT game with marginals corresponding to any valid assignment $f \in \{0, 1/2, 1\}$. 

\end{proof}

\begin{figure}[htbp]
\centering
\includegraphics[width=0.9\textwidth]{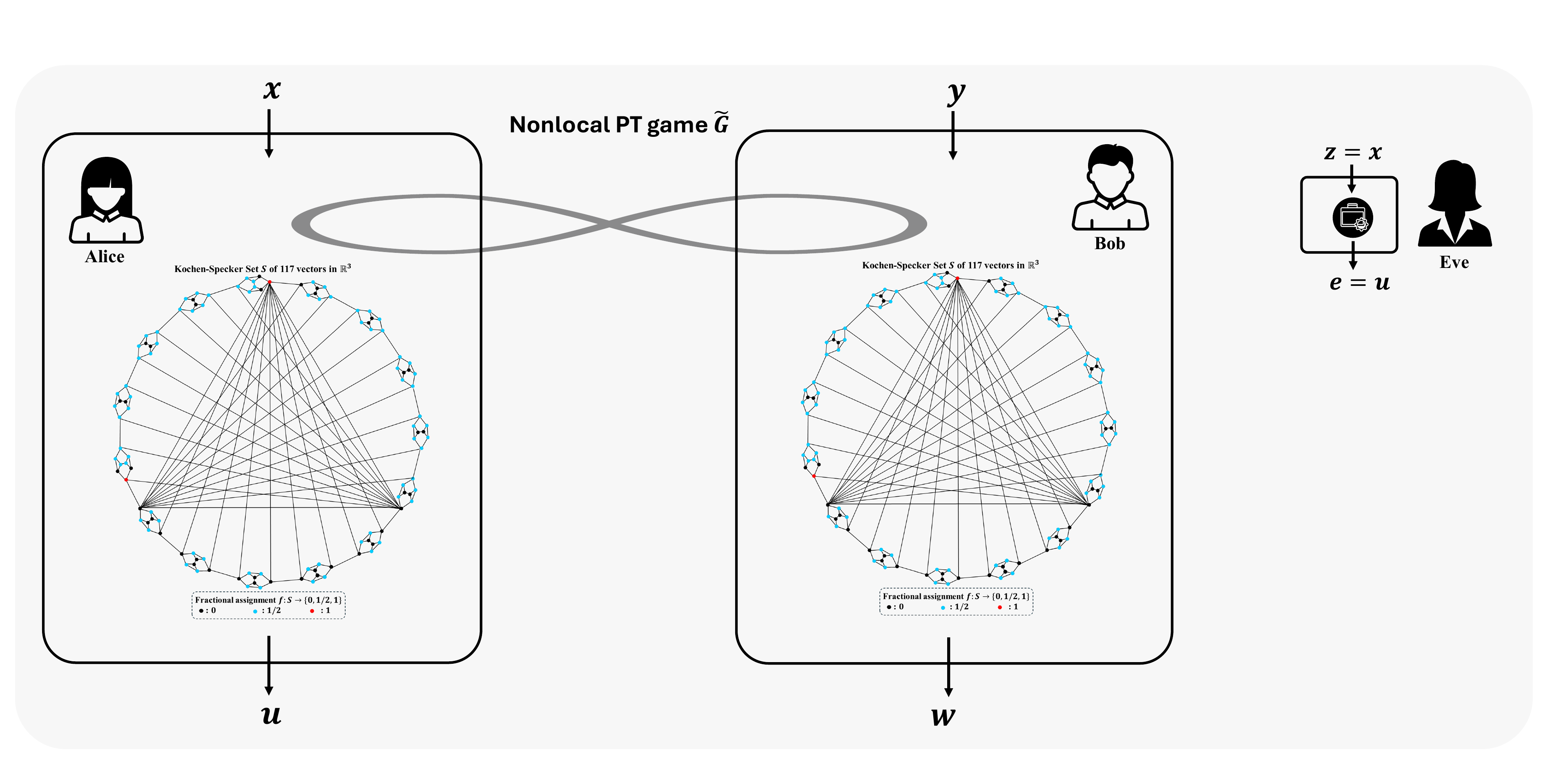}
\label{fig:tripartitedet}
\centering
\caption{The no-signalling strategy is illustrated. The bipartite PT game corresponds to a Kochen-Specker set such as the one illustrated by the orthogonality graph of the projectors. As shown in Lemma \ref{lem:localf-strat} the KS sets admit a $\{0,1/2,1\}$-coloring indicated by the red, black and blue dots. Lemma \ref{lem:bipNS-from-localf} shows how to construct a bipartite Alice-Bob behavior with these marginals. And finally, these bipartite behaviors are used to construct a tripartite no-signalling strategy in Thm. \ref{thm:PT-guessinggame}.  }
\end{figure}

Now, we proceed to construct the tripartite no-signalling strategy $\{P_{\mathsf{U}\mathsf{W}\mathsf{E}|\mathsf{X}\mathsf{Y}\mathsf{Z}}(\mathsf{u}, \mathsf{w}, \mathsf{e}| \mathsf{x}, \mathsf{y}, \mathsf{z}) \}$
that allows Alice and Bob to achieve the maximum value $\omega_{\mathcal{NS}}(G) = 1$ while at the same time, Eve is able to guess Alice's output perfectly i.e., $\mathsf{e}= \mathsf{u}$ when $\mathsf{x}= \mathsf{z}$ following a strategy outlined in \cite{ACP+16}. The construction is illustrated in Fig. \ref{fig:tripartitedet}.

Consider a fixed but arbitrary value of Alice's input $\mathsf{x}$  whose outcomes Eve wants to guess  and set $\mathsf{z} = \mathsf{x}$. We construct the distribution $P_{\mathsf{U}\mathsf{W}\mathsf{E}|\mathsf{X}\mathsf{Y}\mathsf{Z}}(\mathsf{u}, \mathsf{w}, \mathsf{e}| \mathsf{x}, \mathsf{y}, \mathsf{z} = \mathsf{x})$ as follows: $P_{\mathsf{U}\mathsf{W}\mathsf{E}|\mathsf{X}\mathsf{Y}\mathsf{Z}}(\mathsf{u}, \mathsf{w}, \mathsf{e}| \mathsf{x}, \mathsf{y}, \mathsf{z} = \mathsf{x}) = P_{\mathsf{E}|\mathsf{Z}}(\mathsf{e}|\mathsf{z} = \mathsf{x}) \cdot P_{\mathsf{U}\mathsf{W}|\mathsf{X}\mathsf{Y}\mathsf{Z}\mathsf{E}}(\mathsf{u}, \mathsf{w}| \mathsf{x}, \mathsf{y}, \mathsf{z} = \mathsf{x}, \mathsf{e})$ with the distribution $P_{\mathsf{E}|\mathsf{Z}}(\mathsf{e}|\mathsf{z} = \mathsf{x}) = \frac{1}{d}$ for all $\mathsf{e} \in \mathsf{x}$ and the conditional distribution $P_{\mathsf{U}\mathsf{W}|\mathsf{X}\mathsf{Y}\mathsf{Z}\mathsf{E}}(\mathsf{u}, \mathsf{w}| \mathsf{x}, \mathsf{y}, \mathsf{z} = \mathsf{x}, \mathsf{e}) = P^{(f_{e})}_{\mathsf{U}\mathsf{W}|\mathsf{X}\mathsf{Y}}$ from Lemma \ref{lem:bipNS-from-localf}. Here, for any $\mathsf{e}\in\mathsf{x}$,  $f_{e} : S \rightarrow \{0,\frac{1}{2},1\}$ is valid assignment, satisfying the normalization constraints and $f_e(|e\>)=1$.
Evidently when $\mathsf{u} = \mathsf{e} \in \mathsf{x}$ we have that $f_e(|u\>) = 1$ so that we obtain $P_{\mathsf{U}\mathsf{W}\mathsf{E}|\mathsf{X}\mathsf{Y}\mathsf{Z}}(\mathsf{u} = \mathsf{e}, \mathsf{w}, \mathsf{e}| \mathsf{x}, \mathsf{y}, \mathsf{z} = \mathsf{x}) = 1$ or in other words, Eve's output precisely equals Alice's output when Eve chooses her input as $\mathsf{z} = \mathsf{x}$. For the same $\mathsf{x}$ and $\mathsf{z} \neq \mathsf{x}$, Eve is not required to guess Alice's outputs, so one can simply consider the product behavior with the same marginals $P_{\mathsf{U}\mathsf{W}|\mathsf{X}\mathsf{Y}}$ and uniform $P_{\mathsf{E}|\mathsf{Z}}$. 

Now the $|\mathsf{X}|$ different behaviors defined by each $\mathsf{x} \in \mathsf{X}$ can be combined to give the required tripartite behavior $P_{\mathsf{U}\mathsf{W}\mathsf{E}|\mathsf{X}\mathsf{Y}\mathsf{Z}}$. By construction, this tripartite behavior is no-signalling and thus represents a valid attack by a no-signalling eavesdropper. Given knowledge of Alice's input $\mathsf{x}$, Eve can a posteriori choose her input to the device as $\mathsf{z} = \mathsf{x}$ and output her guess as the corresponding outcome $\mathsf{e}$ which by construction allows to perfectly guess Alice's output $\mathsf{u} = \mathsf{e}$. 

\end{proof}

\section{Min-Entropy against Quantum and No-Signalling Adversaries for the Chain Inequality}
\label{sec:ChainIneq-GuessProb}

In this section, we establish a tight analytical bound on the guessing probability as a function of the violation of the chained inequality $I_3$ against quantum adversaries to compare it with the no-signalling case considered earlier. We note that while the (smooth) conditional von Neumann entropy is the quantity considered in proofs of security against quantum adversaries, this entropy can also be bounded from below by the min-entropy. 


We consider the guessing probability $P_g$ defined as the maximal probability with which a quantum adversary, who possesses complete access to the underlying quantum system, can correctly predict the outcome of a given measurement setting of the device. In this scenario, the adversary performs a measurement $E$ on her system and obtains an outcome $e$ to guess Alice's measurement outcome $a$ corresponding to her setting $A_0$. Then, the guessing probability is expressed as
\begin{equation} \label{pg_def}
P_g(A_0|E):=\sum_a P_{AE}(a,e=a|A_0,E).
\end{equation}
and maximized under the constraint that Alice and Bob observe a particular specified value for the chain inequality $I_3$ with $m=3$ inputs per player. As noted earlier, the chain inequality with $m=3$ has been found to achieve even higher key rates (numerically) \cite{MPA11} than the CHSH inequality against quantum adversaries, although analytical proofs of this statement have not been established thus far. As such, our analytical result in this Appendix is of interest on its own right for security against quantum adversaries. To achieve the result, we first introduce a candidate quantum guessing strategy for Eve, providing a concrete lower bound on Eq.\eqref{pg_def}. Subsequently, by demonstrating that this lower bound precisely matches the upper bound derived via the Navascués-Pironio-Acín (NPA) hierarchy \cite{NPA07, NPA08, PNA10} we establish the optimality of our candidate strategy. Consequently, we arrive at a tight analytical characterization of the guessing probability as a function of the violation of the chained inequality $I_3$.

Let $A_0,A_1,A_2$ denote three binary-outcome measurements on Alice’s subsystem, and let $B_0,B_1,B_2$ denote three binary-outcome measurements on Bob’s subsystem.  
The three-setting chained Bell expression $I_3$ is defined as
\begin{equation}\label{I3_def}
    I_3:=\<A_0B_0\>+\<A_1B_0\>+\<A_1B_1\>+\<A_2B_1\>+\<A_2B_2\>-\<A_0B_2\>,
\end{equation}
where $\<O\>:=\Tr(O\rho_{AB})$ is the expectation value of $O$ with respect to the state $\rho_{AB}$ shared by Alice and Bob. Within the DI picture we attribute the choice of the joint state $\rho_{AB}$ together with all local observables $A_i$ and $B_j$ to the adversary Eve. Our candidate guessing strategy for Eve is as follows. Eve prepares the pure three-qubit state $|\Psi\>_{ABE}$ among the devices and her own system that
\begin{equation}
    |\Psi\>_{ABE}=\frac{1}{\sqrt{2}}\left(|0\>_A|0\>_{B}|e_0\>_{E}+|1\>_A|1\>_{B}|e_1\>_{E}\right)
\end{equation}
where $|e_0\>,|e_1\>$ are normalized qubit states with overlap $\<e_0|e_1\>=\sin(\frac{\theta}{2})$ for parameter $\theta\in[0,\pi)$.

The local measurement operators are prepared as
\begin{equation}
    \begin{array}{ll}
        A_0=\sigma_z,\qquad & \; \; \; \; B_0=-\cos\varphi_b \sigma_z +\sin\varphi_b \sigma_x,\\
        A_1=-\cos\varphi_a \sigma_z +\sin\varphi_a \sigma_x,\qquad  & \; \; \; \; B_1=\sigma_x,\\
         A_2=\cos\varphi_a \sigma_z +\sin\varphi_a \sigma_x,\qquad &\; \; \; \; B_2=\cos\varphi_b \sigma_z +\sin\varphi_b \sigma_x,\\
    \end{array}
\end{equation}
with $\sigma_z$ and $\sigma_x$ the Pauli operators.  
The angles $\varphi_a,\varphi_b\in\bigl(\tfrac{\pi}{2},\pi\bigr)$ are related to $\theta$ through
\begin{equation}\label{angle_relation}
    \begin{split}
        \cos\varphi_a&=\frac{1-\cos\theta-\sqrt{3+\cos^2\theta}}{1+\cos\theta},\\
        \sin\varphi_b&=\frac{2-\sqrt{3+\cos^2\theta}}{1+\cos\theta}.
    \end{split}
\end{equation}
The reduced state of Alice and Bob is
\begin{equation}
    \rho_{AB}=\frac{\mathbb{I}\otimes \mathbb{I}}{4}+\frac{\sigma_z\otimes \sigma_z}{4}+\sin{\left(\frac{\theta}{2} \right)}\frac{\sigma_x\otimes \sigma_x}{4}-\sin{\left(\frac{\theta}{2}\right)}\frac{\sigma_y\otimes \sigma_y}{4}.
\end{equation}
A direct calculation then yields the quantum value $w_{\mathcal{Q}}$ of the chained expression $I_3$ \eqref{I3_def} achieved by this strategy as: 
\begin{equation}\label{wq}
\begin{split}
    w_{\mathcal{Q}}&=\Tr\left[\rho_{AB} \left(A_0B_0+A_1B_0+A_1B_1+A_2B_1+A_2B_2-A_0B_2\right)\right]\\
    &= \Tr [\rho_{AB} A_0(B_0-B2)]+\Tr [\rho_{AB} (A_1+A_2)B_1]+\Tr [\rho_{AB}(A_1B_0+A_2B_2)]\\
    &=-2\cos\varphi_b \Tr [\rho_{AB} \sigma_z\otimes \sigma_z]+2\sin\varphi_b \Tr [\rho_{AB}\sigma_x\otimes\sigma_x]+\Tr [\rho_{AB} (2\cos\varphi_a \cos\varphi_b \sigma_z\otimes \sigma_z+ 2\sin\varphi_a\sin\varphi_b\sigma_x\otimes \sigma_x] \\
    &=2\cos\varphi_b(\cos\varphi_a-1)+2\sin\varphi_a(\sin\varphi_b+1)\sin(\frac{\theta}{2})\\
    &=\frac{\sqrt{2\left(-3+\cos\theta+2\sqrt{3+\cos^2\theta}\right)}\left(3-\cos\theta+\sqrt{3+\cos^2\theta}\right)}{1+\cos\theta}
\end{split}
\end{equation}
where, to establish the last line, we have repeatedly used the following useful relations between $\varphi_a,\varphi_b$ obtained from Eq.~\eqref{angle_relation}:
\begin{equation}\label{relation}
\begin{split}
    &\sin\varphi_b-\cos{\varphi_a}=1;\\
    &2\cos{\varphi_a}+1+\cos^2{\varphi_b} \frac{1+\cos\theta}{2}=0;\\
    & 2\sin{\varphi_b}-1+\sin^2{\varphi_a} \frac{1+\cos\theta}{1-\cos\theta}=0;\\
    & \sin\varphi_a=-\cos \varphi_b\sin(\frac{\theta}{2}).
\end{split}
\end{equation}

On the other hand, after Alice measures $A_0$, the correlation between her classical outcome and Eve’s system is described by the classical–quantum state:
\begin{equation}
    \frac{1}{2}[0]_A\otimes |e_0\>\<e_0|+\frac{1}{2}[1]_A\otimes |e_1\>\<e_1|
\end{equation}
where $[0]$ and $[1]$ label the two possible outcomes of $A_0$. 
Consequently, Eve’s optimal guessing strategy in this circumstance is to perform the measurement that can optimally ``unambiguous discriminate'' the quantum state $|e_0\>,|e_1\>$ with the uniform prior probability. According to the seminal result by Helstrom~\cite{Helstrom69}, the minimum error probability for states discrimination (for $|e_0\>,|e_1\>$ with uniform prior probability) is achieved by the two-outcome POVM $\{E_0,\,E_1=\mathbb{I}_2-E_0\}$, where $E_0$ is the projector on the positive eigenspace of the operator $\frac{1}{2}|e_0\>\<e_0|-\frac{1}{2}|e_1\>\<e_1|$. Furthermore, the resulting error probability is
\begin{equation}
    P_{err}=\frac{1}{2}\left(1-\sqrt{1-
    |\<e_0|e_1\>|^2}\right).
\end{equation}
Hence the optimal guessing probability reads
\begin{equation}\label{pg}
    P_g(A_0|E)=1-P_{err}=\frac{1}{2}\left(1+\sqrt{1-
    |\<e_0|e_1\>|^2}\right)=\frac{1}{2}\left(1+\cos(\frac{\theta}{2})\right).
\end{equation}

This Eq.~\eqref{pg}, combining with the quantum violation $w_{\mathcal{Q}}$ of $I_3$ in Eq.~\eqref{wq} yields an explicit relation between the guessing probability obtained from this candidate guessing strategy and the value of the chained Bell expression $I_3$.  Eliminating the parameter $\theta$ gives $P_g(A_0|E)$ as a function of the quantum violation $w_{\mathcal{Q}}$:
\begin{equation}\label{pg_ana}
\begin{split}
    P_g(A_0|E)&=\frac{1}{2}\left(1+\sqrt{\frac{1+f(w_{\mathcal{Q}})}{2}}\right)\\
\end{split}
\end{equation}
where $f(w_{\mathcal{Q}})$ is the real root $x$ of the cubic equation:
\begin{equation}
    8w_{\mathcal{Q}}^2 x^3 + (w_{\mathcal{Q}}^4-432)x^2+(2w_{\mathcal{Q}}^4-72 w_{\mathcal{Q}}^2+864)x+(w_{\mathcal{Q}}^4-432)=0.
\end{equation}
More explicitly we have:
\begin{equation}
\begin{split}
  f(w_{\mathcal{Q}})&=(-\frac{q}{2}+\sqrt{\Delta})^{1/3}+(-\frac{q}{2}-\sqrt{\Delta})^{1/3}-\frac{b}{3a},\\
 \text{where}  \quad a&=8w_Q^{2},
    b=w_{\mathcal{Q}}^{4}-432,
    c=2w_{\mathcal{Q}}^{4}-72 w_{\mathcal{Q}}^{2}+864,
    d=w_{\mathcal{Q}}^{4}-432,\\
     \quad p&=\frac{3ac-b^{2}}{3a^{2}},
    q=\frac{2b^{3}-9abc+27a^{2}d}{27a^{3}},
    \Delta=\left(\frac{q}{2}\right)^{2}+\left(\frac{p}{3}\right)^{3}.
\end{split}
\end{equation}

Finally, note that in the above analysis, our parameterization excludes $\theta=\pi$.  
At $\theta=\pi$ we have $|e_0\rangle=|e_1\rangle$ and the global state factorizes as $|\Psi\rangle_{ABE}=|\Psi^{+}\rangle_{AB}|e_0\rangle_{E}$. In this extremal case Alice and Bob share a pure maximally entangled state that is decoupled with Eve. In this case, it's known that the maximum quantum value of $3\sqrt{3}$ for the chained Bell expression $I_3$ is achieved, while Eve's optimal guessing probability is $\frac{1}{2}$.

A numerical upper bound on the guessing probability can be obtained based on the Navascués-Pironio-Acín (NPA) hierarchy~\cite{NPA07, NPA08, PNA10}, which is a sequence of outer approximations to the quantum correlation set.  Let $\tilde{\mathcal{Q}}_{k}$ denote the level-$k$ relaxation, increasing $k$ yields progressively tighter approximations to the true quantum correlation set $\mathcal{Q}$.  For a fixed observed violation $w_{\mathcal{Q}}$ of the chained inequality $I_{3}$, the corresponding upper bound on the guessing probability is obtained by solving the semidefinite programme
\begin{equation}\label{sdp_bound}
  \begin{split}
    \max_{\{p(a,b,e|x,y,E)\}} \quad & P_g(A_0|E)\\
     s.t.\qquad & \sum_{a,b,x,y} c_{a,b,x,y} p(a,b|x,y) \geq w_{\mathcal{Q}}\\
           & \{p(a,b,e|x,y,E)\} \in\tilde{\mathcal{Q}}_k\\
  \end{split}
\end{equation}
where $\sum_{a,b,x,y} c_{a,b,x,y}\,p(a,b|x,y)$ denotes $I_{3}$ written out in terms of probabilities.

Figure~\ref{fig_pg} displays the analytical curve of Eq.~\eqref{pg_ana}, obtained from our candidate strategy, and also displays the SDP upper bound derived from Eq.~\eqref{sdp_bound} at level $k=2$.  As evident from their comparison in Figure~\ref{fig_pg}, the analytical curve of Eq.~\eqref{pg_ana} is tight, confirming that the candidate strategy is optimal.  

\begin{figure}[htbp]
\centering
\subfigure[]{
\begin{minipage}[t]{0.48\linewidth}
\centering
\includegraphics[width=\textwidth]{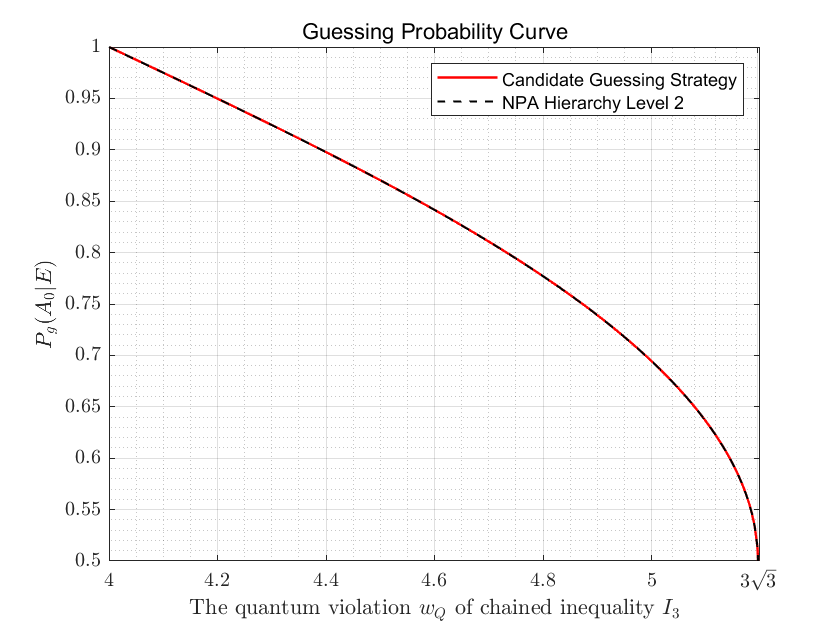}
\label{fig_pg}
\end{minipage}%
}%
\subfigure[]{
\begin{minipage}[t]{0.48\linewidth}
\centering
\includegraphics[width=\textwidth]{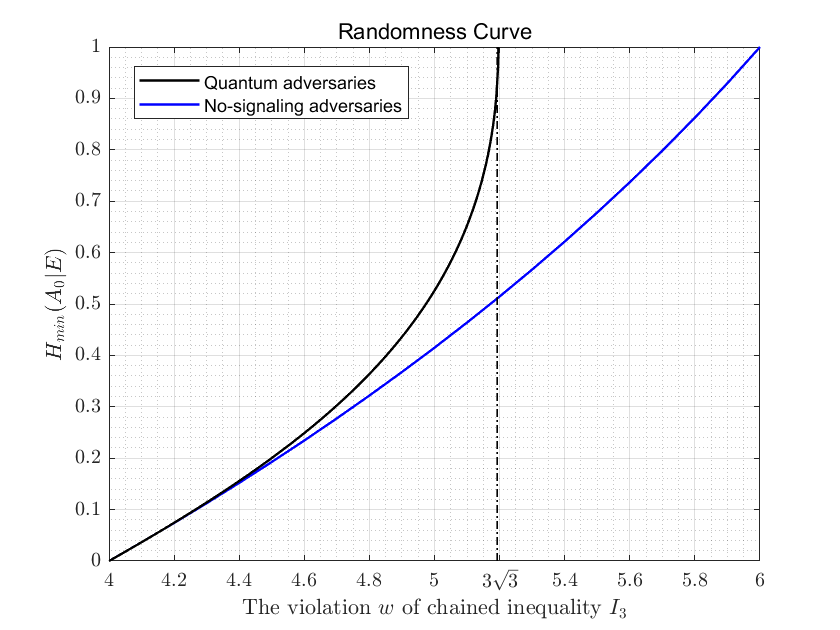}
\label{fig_hmin}
\end{minipage}%
}
\centering
\caption{(a) The guessing probability $P_g(A_0|E)$ versus the quantum violation $w_Q$ of chained inequality $I_3$. (b) The min-entropy $H_{min}(A_0|E)$ as a function of the value $w$ of the chained Bell expression $I_3$ is plotted against both quantum (black, \eqref{pg_ana}) and no-signaling (blue, \eqref{eq:pg_chainns}) adversaries.}
\end{figure}

The guessing probability against a no-signaling adversary can be computed via a linear program and proven through strong duality. The guessing probability against a no-signaling adversary is computed as a function of the value $w_{\mathcal{NS}} \in [4,6]$ (where $4$ is the maximum classical value and $6$ is the maximum no-signaling value) of the chained Bell expression $I_3$ in \eqref{I3_def}  to be 
\begin{equation}
\label{eq:pg_chainns}
    P_g(A_0|E)=-\frac{w_{\mathcal{NS}}}{4}+2.
\end{equation}

Finally, the comparison of the min-entropy $H_{min}$ against quantum and no-signaling adversaries, as a function of the violation $w$ of the chained inequality $I_3$, is plotted in Fig.~\ref{fig_hmin}. The toolkit Moment~\cite{GA24}, the modelers YALMIP~\cite{Lofberg04} and CVX~\cite{CVX, GB08}, and the solver MOSEK~\cite{Mosek} are used for SDP and LP calculation. A generalisation of the guessing probability expressions \eqref{pg_ana} and \eqref{eq:pg_chainns} for general $m \geq 3$ inputs per player and the accumulation of entropy \cite{AFDFRV18, DFR20} for the chained Bell expression in experimental DI randomness certification against quantum and no-signalling adversaries is pursued in forthcoming work.


\end{document}